\documentclass{amsart}
\usepackage {CJKutf8}
\usepackage{amsmath}
\usepackage{amssymb}
\usepackage{amsthm}
\usepackage{hyperref}

\usepackage{todonotes}

\usepackage{enumitem}

\usepackage{tikz}

\usepackage{float}

\usepackage[numbers]{natbib}

\theoremstyle{plain}%
\newtheorem{theorem}{Theorem}
\newtheorem{proposition}[theorem]{Proposition}%
\newtheorem{lemma}[theorem]{Lemma}
\newtheorem{corollary}[theorem]{Corollary}

\theoremstyle{remark}%
\newtheorem{remark}{Remark}%

\theoremstyle{definition}%
\numberwithin{equation}{section}
\numberwithin{theorem}{section}

	\newcommand{\R}{\mathbb{R}}

	\newcommand{\cB}{\mathcal{B}}
	\newcommand{\cD}{\mathcal{D}}

	\newcommand{\cF}{\mathcal{F}}
	\newcommand{\cG}{\mathcal{G}}
	\newcommand{\cH}{\mathcal{H}}

	\newcommand{\cL}{\mathcal{L}}
	\newcommand{\cN}{\mathcal{N}} 
	\newcommand{\N}{\mathcal{N}}       
	
	\newcommand{\cR}{\mathcal{R}}
	\newcommand{\cS}{\mathcal{S}}

	\newcommand{\cX}{\mathcal{X}}

	\newcommand{\nc}{\newcommand}

	\newcommand{\Del}{\Delta}
	
	\nc{\bP}{\bar{P}}
	\nc{\bQ}{\bar{Q}}

	\newcommand{\Cb}{\mathbb{C}}

	\newcommand{\Rb}{\mathbb{R}}
	
	\newcommand{\Hc}{\mathcal{H}}

	\newcommand{\si}{\sigma}
	\newcommand{\eps}{\epsilon}
	
	\newcommand{\g}{\gamma}
	\newcommand{\al}{\alpha}
	
	\newcommand{\Si}{\Sigma}

	\DeclareMathOperator{\supp}{supp}

	\DeclareMathOperator{\dist}{dist}

	\newcommand{\grad}{\nabla}
	\newcommand{\Lap}{\Delta}
	\newcommand{\di}{\partial}

	\DeclareMathOperator{\Tr}{Tr}

	\nc{\ran}{\rangle}
	\nc{\lan}{\langle}

	\newcommand{\ra}{\rightarrow}

	\renewcommand{\Im}{\mathrm{Im}} 

	\newcommand{\tr}{\mathrm{Tr}}

	\nc{\bfone}{{\bf 1}}
	\nc{\1}{{\bf 1}}
	
	\newcommand{\p}{\partial}

	\newcommand{\n}{\nabla}

	\newcommand{\DETAILS}[1]{}

	\newcommand{\one}{\ensuremath{\mathbf{1}}}
	
	\newcommand{\br}[1]{\left\langle#1\right\rangle}

	\nc{\den}{\text{den}}
	\nc{\ex}{\text{xc}}
	\nc{\Ex}{\text{Xc}}

	\DeclareMathOperator{\Prob}{Prob}
	\DeclareMathOperator{\Ad}{ad}
	\newcommand{\ad}[3]{\mathrm{ad}^{#1}_{#2}(#3)}

	\newcommand{\Rem}{\mathrm{Rem}}

	\newcommand{\cp}{\mathrm{c}}

	\newcommand{\omn}[1]{\Tr\del{#1\rho_0}}

	\newcommand{\abs}[1]{\ensuremath{\left\lvert#1\right\rvert}}
	\newcommand{\norm}[1]{\ensuremath{\left\lVert#1\right\rVert}}
	\newcommand{\sbr}[1]{\left[#1\right]}
	\newcommand{\Set}[1]{\left\{#1\right\}}

	\newcommand{\md}[6]{\ensuremath{
			\ifinner
			\tfrac{\partial{^{#2}}#1}{\partial{#3^{#4}}\partial{#5^{#6}}}
			\else
			\tfrac{\partial{^{#2}}#1}{\partial{#3^{#4}}\partial{#5^{#6}}}
			\fi
	}}
	\newcommand{\del}[1]{\left(#1\right)}
	\newcommand{\thmref}[1]{Theorem~\ref{#1}}

	\newcommand{\secref}[1]{Section~\ref{#1}}
	\newcommand{\lemref}[1]{Lemma~\ref{#1}}
	\newcommand{\propref}[1]{Proposition~\ref{#1}}
	\newcommand{\remref}[1]{Remark~\ref{#1}}
	\newcommand{\figref}[1]{Figure~\ref{#1}}
	\newcommand{\corref}[1]{Corollary~\ref{#1}}

	\usepackage{changes}
	\definechangesauthor[name={DO}, color=magenta]{do}
	\definechangesauthor[name={MS}, color=green]{ms}
	\definechangesauthor[name={JZ}, color=blue]{jz}
	\definechangesauthor[name={JF}, color=red]{jf}
	
	\usepackage{soul}
	
	\usepackage{soul}
	\setstcolor{red}

	\definecolor{green}{rgb}{0.0, 0.5, 0.5}
	\definecolor{yellow}{rgb}{0.5, 0.5, 0}
	\definecolor{lgray}{gray}{0.9}
	\definecolor{llgray}{gray}{0.95}
	\definecolor{lllgray}{gray}{0.975}

\begin{document}
		
		\title[Light cones for open quantum systems]{Light cones for open quantum systems
		}		
		
	 \author[S.~Breteaux]{S\'ebastien Breteaux}
\address{Institut Elie Cartan de Lorraine, Universit\'e de Lorraine, 57045 Metz Cedex 1, France}
 \email{sebastien.breteaux@univ-lorraine.fr}

 \author[J.~Faupin]{J\'er\'emy Faupin}
\address{Institut Elie Cartan de Lorraine, Universit\'e de Lorraine, 57045 Metz Cedex 1, France}
 \email{jeremy.faupin@univ-lorraine.fr}

		\author[M.~Lemm]{Marius Lemm}
		\address{Department of Mathematics, University of T\"ubingen, 72076 T\"ubingen, Germany}
		\email{marius.lemm@uni-tuebingen.de}
		
		\author[D.H.~Ouyang]{Dong Hao Ou Yang}
		\address{Department of Mathematics, University of Toronto, Toronto, ON M5S 2E4, Canada }
		\email{donghao.ouyang@mail.utoronto.ca}

		\author[I.~M.~Sigal]{Israel Michael Sigal}
		\address{Department of Mathematics, University of Toronto, Toronto, ON M5S 2E4, Canada }
		\email{im.sigal@utoronto.ca}

		\author[J.~Zhang]{Jingxuan Zhang
		}
		\address{Department of Mathematical Sciences, University of Copenhagen, Copenhagen 2100, Denmark}
		\email{jingxuan.zhang@math.ku.dk}

		\date{March 15, 2023}
		\begin{abstract}
			We consider Markovian open quantum dynamics (MOQD). We show that, up to small-probability tails, the supports of quantum states evolving under such dynamics propagate with finite speed in any finite-energy subspace.

			More precisely, we prove that if the initial quantum state is localized in space, then any finite-energy part of the solution of the von Neumann-Lindblad equation is approximately localized inside an energy-dependent light cone. We also obtain an explicit upper bound for the slope of this light cone. 
			
		\end{abstract}
		
		\maketitle

		\section{Introduction}
		While non-relativistic quantum theory does not possess the strict light cone of relativistic theories, it has been shown in many contexts that its dynamics nonetheless exhibits a maximal speed bound up to small-probability leakage. By analogy, one speaks of a (system-dependent) \textit{light cone} also in these cases. Existence of such light cones has been rigorously derived in standard QM \cite{APSS,HerbstSkib,SigSof,Skib}, for non-relativistic QED models \cite{BFS}, and for nonlinear Schrödinger equations \cite{AFPS}. Famously, Lieb and Robinson \cite{LR} first derived the existence of light cones in quantum spin systems. Their eponymous Lieb-Robinson bounds have developed into an extremely active research area starting in the early 2000s \cite{H04,H07,HW,NRSS,NS1} and continues to grow in scope, e.g., with recent extensions to lattice fermions \cite{GGC,NSY1}, lattice bosons \cite{FLS1,FLS2,KSV,SHOE, SZ,WH,YL} and long-range interactions \cite{Fossetal,GGC,Tranetal}. The existence of a maximal speed bound in a quantum theory is a fundamental statement about its non-equilibrium properties which serves as the backbone of many proofs. For instance, it played an essential role in scattering theory \cite{Der,SigSof2} and, in quantum information theory Lieb-Robinson bounds were used to prove the celebrated area law for entanglement entropy \cite{H04} and bounds on quantum state transfer \cite{EW}. They are also central to  the notion of quantum phase defined via quasi-adiabatic continuation \cite{HW,NSY2}.

		In this paper, we consider quantum particles governed by the Schrödinger operator $H=-\Delta+V$ that interact with an environment. We show that the corresponding Markovian open quantum dynamics (MOQD) exhibit an energy-dependent light cone, i.e., initially localized states propagate at most with a maximal speed. Previous results about maximal speed bounds of MOQD either concerned lattice systems (where the mechanism for maximal speed is different \cite{NVZ,Pou}) or it excluded the most interesting case when the Hamiltonian $H$ is a standard Schrödinger operator \cite{BFLS}. In this paper, we resolve this question and show that coupling quantum-mechanical particles to an environment cannot lead to acceleration of any finite-energy portion. For this purpose, we develop microlocalization techniques involving functions of noncommuting operators $H$ and $x_j$. To fix ideas, we work on $L^2(\mathbb{R}^d)$ but we expect that our approach could be extended to abstract Hilbert space with abstract noncommuting self-adjoint operators $H$ and $x_j$.
				

		\medskip
		
		\subsection{Setup and main result}
		\label{ref:result}
		We study the long-time behaviour of solutions to the von Neumann-Lindblad (vNL) equation:  
		\begin{align}\label{vNLeq}
			&\frac{\partial\rho_t}{\partial t}=-i[H,\rho_t]+\frac12\sum_{j\geq 1}\big([W_{j},\rho_t W_{j}^{*}]+[W_{j}\rho_t,W_{j}^{*}]\big). 
		\end{align}
		Here $\rho_t,\,t\ge0$ is a family of density operators (i.e.~non-negative-definite operators with unit trace) on a Hilbert space $\cH$,   $H$ is the quantum Hamiltonian, a self-adjoint operator on $\mathcal{H}$,  and the $\{W_{j}\}$ are bounded operators, arising from interaction with the environment.

		We show that, for any $E$, there exists $\kappa= \kappa(E)>0$ such that, for any  initial condition $\rho_0$ localized in $X\subset\R^d$ and for any $c>\kappa$, the probability that the system in the state $\rho_t$ is localized in $\cH_E\cap X^\cp_{ct}$ is arbitrarily small, asymptotically as $t\to\infty$, where $\cH_E$ is the spectral subspace 
		\begin{equation*}
			\cH_E:=\{H\le E\}\equiv\mathrm{Ran}(\one_{(-\infty,E]}(H))
		\end{equation*}
		and $X^\cp_{ct}=\R^d\setminus X_{ct}$ with \begin{equation}\label{Xct}
			X_{ct}\equiv  \Set{x\in\Rb^d:d_X(x)\le ct}
		\end{equation}  the light cone corresponding to a smoothed out distance function $d_X(\cdot)$ defined in \eqref{dX} below. Put differently, there exists an energy-dependent light cone for \eqref{vNLeq} with slope $\kappa$.\\

		Throughout this article, we let $\mathcal{H}=L^2(\mathbb{R}^d)$, $d\ge1$. 
		We make no distinction in our notation between functions and the operators of multiplication defined by those functions. 		For an operator $A$ on $\cH$,  denote by $\mathcal{D}(A)\subset \cH$ the domain of $A$.	
		
		We now set out the main assumptions in this paper. 
		We take the Hamiltonian 
		$H$ in \eqref{vNLeq} to be the standard Schr\"odinger operator,
		\begin{equation}
			\label{H}
			H=-\Lap+V(x), \quad V:\Rb^d\to\Rb. 
		\end{equation} 
		Then, for some fixed integer $n\ge1$,  we assume
		\begin{enumerate}[label=\textbf{(H)}] 
		\item\label{H1}  There exist $\rho>0$ and $C>0$ such that 
		\begin{equation}
			\label{V-cond} 
			\abs{\di^\al V(x)}\le C \br{x}^{-\abs{\al}-\rho}\quad (x\in\Rb^d,\,0\le\abs{\al}\le n).
		\end{equation}
		Here and below, we write $\br{\cdot}=\sqrt{1+\abs{\cdot}^2}.$
	\end{enumerate}	
	
	\begin{remark}
		If $V$ satisfies \ref{H1}, then it is bounded and therefore $H$ is self-adjoint on $\cD(-\Lap)$ (see e.g. \cite{CFKS}) and bounded from below. 
	\end{remark}

	For the operators $W_j,\,j\geq 1$ in \eqref{vNLeq}, we assume, for the same integer $n\ge1$ as in \ref{H1}:
	\begin{enumerate}[label=\textbf{(W\arabic*)}]
		\item\label{W1} For all integers $j\geq 1$, $W_j\in\cB(\cH)$ and the series $\sum_{j=1}^\infty W_{j}^{*}W_{j}$  converges strongly in $\cB(\cH)$ (and consequently,  $\sum_{j=1}^\infty W_{j}^{*}W_{j}\in\cB(\cH)$);
		

		\item\label{W2} Let $C_A=\Ad_A: B\ra [A, B]$ and $p_q=-i\partial_{x_q}$. Then, for every $1\le q\le d$,		\begin{align}\label{W2-cond'}
			\sum_{j=1}^\infty\sum_{\substack{\sum(k_i+\ell_i)=n+1\\ k_i,\,\ell_i\ge0}}\|\prod_{i} \big[(\br x C_{p_q})^{k_i}C_{x_q}^{\ell_i} W_j\big]\|^2<\infty.
		\end{align}

	\end{enumerate}

	{		\begin{remark} Assumptions \ref{W1} and \ref{W2} can be ensured for example by taking the $W_j$'s to be suitable pseudodifferential operators. See also \cite[Section 1.4]{BFLS} and \cite[Section 4]{FFFS}
	\end{remark}}
	
	\begin{remark}
		Let $\cS_1$ stand for the Schatten space of trace-class operators. 
		Conditions \ref{H1} and \ref{W1} guarantee global well-posedness for \eqref{vNLeq} in the space  
		\begin{align}\label{domain_def}
			\cD&:=\{\rho\in\cS_1\mid \rho\cD(H)\subset\cD(H)\text{ and }[H,\rho]\in \cS_1\},
		\end{align}
		see below.
	\end{remark}
	
	\medskip

	For each subset $X\subset\Rb^d$, let $X^\cp:=\Rb^d\setminus X$ 	and $\chi_X^\sharp$ stand for the characteristic function of $X$.	
	The main result of this paper is the following:

	\begin{theorem}[{Main result}] \label{thm1}  
		Suppose Assumptions \ref{H1} and \ref{W1}--\ref{W2} hold. Let $X\subset \Rb^d$ be a bounded and closed subset.   Suppose $\rho_0\in\cD$ (see \eqref{domain_def}) is supported in $X$ in the sense that  \begin{equation}\label{lam-cond}
			\Tr(\chi_{X^\cp}^\sharp\rho_0)=0.
		\end{equation} 
		Then \eqref{vNLeq} has a unique solution $\rho_t\in\cD,\,t\ge0$, and for any $E\in\si(H)$ and  $c>\kappa
		$
		with $\kappa$ as in \eqref{kappa}, this solution satisfies
		\begin{equation}
			\label{1.2}
			\Tr(g(H)\chi^\sharp_{X_{ct}^\cp}g(H)\rho_t)\le  
			C_{n,E} t^{-n},
		\end{equation}
		for all  {$t> 0$} 

		and all smooth cutoff functions $g$ with $\supp(g)\subset (-\infty,E]$ and $0\le g\le1$, where $X_{ct}^\cp\equiv (X_{ct})^\cp$ and $C_{n,E}$ is a positive constant depending on $n$ and $E$.

	\end{theorem}

	\begin{remark}
		For the energy-dependent speed $\kappa$ defined in \eqref{kappa}, we have the following estimate:
		\begin{align} \label{kappa-est}
			\kappa\le C(1+|E|)^{1/2} \text{ for some fixed $C>0$ and all }X\subset \Rb^d,E\in\Rb. 
		\end{align}
		Moreover, the constant $C_{n,E}$ in \eqref{1.2} grows polynomially with $E$.
	\end{remark}

	{\thmref{thm1} solves an open problem from \cite{BFLS}, namely, to derive a light cone for MOQD when the Hamiltonians is a standard Schr\"odinger operator $-\Delta +V$ (a situation not covered by the methods in \cite{BFLS}).}
	
	\thmref{thm1} is proved in  \secref{sec:main-thm-pf}. 
	\thmref{thm1} implies that ``microlocally'' the propagation speed for \eqref{vNLeq} is finite, and yields an upper bound for the maximal speed of propagation of initially localized states. Indeed, 		define the  probability
	\begin{equation}\label{1.9s}
		\Prob_{\rho_t, E} \del{Y}:=\Tr(g_E(H)\chi^\sharp_Yg_E(H)\rho_t)
	\end{equation}
	for the system  in the state $\rho_t$ to be in the part of the state (phase) space where $x\in Y$ and $ H\le E$. 
	With notation \eqref{1.9s} and, recall, $X_{ct}^\cp\equiv (X_{ct})^\cp$, the exterior of the light cone $X_{ct}$ in \eqref{Xct}, \thmref{thm1} says that 
	$$\Prob_{\rho_t, E}(X_{ct}^\cp)\le  {C_{n,E}}t^{-n}.$$
	
	 The constant $C_{n,E}$ in \eqref{1.2} depends on the difference $c-\kappa>0$ (through \eqref{propag-est3} below). For brevity of notation, we do not display the dependence on $c-\kappa$.

	{\subsection{Explicit formula for $\kappa$} }
	{In equations \eqref{gamma}-\eqref{kappa} below, we provide an explicit formula for the number $\kappa$ in \thmref{thm1}. Physically, $\kappa$ bounds the propagation speed (also called ``speed of sound'') in the energy-constrained open quantum system. Naturally, $\kappa$ depends on the system parameters and the energy cutoff.} 
		
	We first introduce some notations. For each closed set $X\subset\R^d$, we define the \textit{smoothed distance function} to $X$, $d_X\in C^\infty(\Rb^d)$ in the following way. Let $\epsilon_0>0$ be a fixed parameter (the estimate \eqref{1.2}, in particular, depends on this arbitrary parameter). Let
	\begin{equation}\label{dX}
		d_X(x)\equiv d_{X,\eps_0}(x)\left\{\begin{aligned}
			&=0,\quad &\dist_X(x)=0,\\
			&\ge 0,\quad \quad &0< \dist_X(x)<c_1 \eps_0,\\
			&=\delta_X(x)-\eps_0,\quad &\dist_X(x)\ge c_1\eps_0,
		\end{aligned}\right.
	\end{equation}
	where $\delta_X\in C^\infty(\Rb^d)$ satisfies $c_1 \dist_X (x) \le \delta_X(x) \le c_2 \dist_X (x)$ for some $c_1,c_2>0$, and
	\begin{align}
		\dist_X^{\abs{\al}-1}(x)\abs{\di^\al d_X(x)}\le C_\alpha \label{dX2}\quad (x\in\Rb^d,\,0\le \abs{\al}),
	\end{align}
	for some absolute constants $C_\alpha>0$.  	
	In one-dimension, such functions are easy to construct, see the schematic diagram \figref{fig:phi}. In any dimension, one can proceed as follows. By the extension theorem of Whitney (see e.g. \cite[Theorem 6.2.2]{Stein70}), there exists a function $\delta_X$ defined in $X^c$ such that 
	\begin{align*}
		&c_1 \dist_X (x) \le \delta_X(x) \le c_2 \dist_X (x) , \quad \text{for all }x\in X^c \\
		&\delta_X \text{ is } C^\infty \text{ in } X^c \text{ and } \dist_X^{\abs{\al}-1}(x) \partial^\alpha\delta_X(x) \le C_\alpha, \quad \text{for all }x\in X^c \text{ and } |\alpha|\ge0,
	\end{align*}
	where $c_1,c_2,C_\alpha$ are positive constants independent of $X$. Let $f_{\eps_0}:\R\to\R$ be a $C^\infty$ function such that $f_{\eps_0}(x)=0$ if $x\le{\eps_0}/2$, and $f_{\eps_0}(x)=x-\eps_0$ if $x\ge\eps_0$. We can then define
	\begin{equation*}
		d_X(x):=f_{\eps_0}(\delta_X(x))
	\end{equation*}
	and verify that it satisfies the conditions above.
	
	\begin{figure}[H]
		\centering
		\begin{tikzpicture}[ scale=4]
			\draw[-] (-1.5,0)--(1,0);
			\node[right] at (1,0) {$\Rb^d$};
			
			\draw[very thick] (-1.5,0)--(-.6,0);
			\node[below] at (-1.2,0) {$X$};
			\draw [fill] (-.6,0) circle [radius=0.02];
			\node[below] at (-.6,0) {$\di X$};
			
			
			\draw[very thick] (-.6,0) [out=0, in=-135,] to (-.1,.2) [out=45, in=-135] to (.5,.8);
			
			\draw[thick,dashed] (-.6,0)-- (.2,.8);
			
			\node[right] at (.5,.8) {$d_X(x)$};
			\node[left] at (.17,.8) {$\dist_X(x)$};
			
		\end{tikzpicture}
		\caption{Schematic diagram illustrating $d_X\equiv d_{X,\eps}$ in \eqref{dX}. }\label{fig:phi}
	\end{figure}
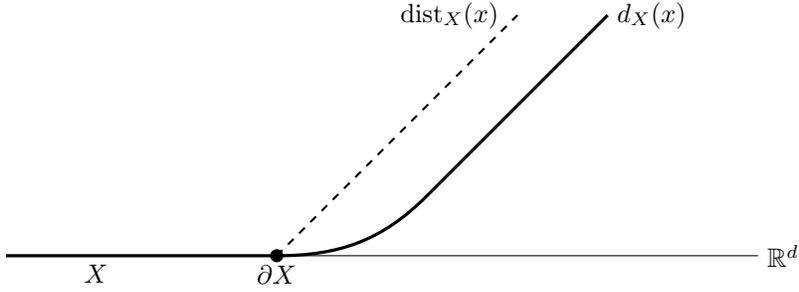

	We fix $E\in\si(H)$ and a function $g\in C^\infty(\Rb)$ satisfying $0\leq g\leq 1$ and, for some small $\eps>0$, 
	\begin{equation}\label{g-cond}
		g(\mu)\equiv 1 \text{ for }\mu\le E-\eps,\quad g(\mu)\equiv 0 \text{ for }\mu\ge E,
	\end{equation}
	and define the \textit{smooth energy cutoff} operator
	\begin{equation}\label{g}
		g:=g(H). 
	\end{equation}
	\begin{remark}\label{remG}
		Since $g(H)=(g\chi^\#_{\si(H)})(H)$, the values of $g$ outside of $\si(H)$ are irrelevant. Since, moreover, $H$ is bounded from below by \ref{H1}, one can always take $g$ to have compact support if needed.
	\end{remark}

	Considering the multiplication operator $d_X$ by the smoothed distance function $d_X(x)$, introduced in \eqref{dX} above, we define 	the spectrally localized distance function
	\begin{equation}\label{gE}
		d_X^E:=gd_Xg\quad\quad\text{ defined on } \quad \{u\in\Hc: g u\in \cD(d_X)\}.
	\end{equation} 
	Now, we define the \textit{energy-dependent velocity operator}
	\begin{align}\label{gamma}
		\g\equiv \g(X,E):=
		i[H,d_X^E]+\frac{1}{2}\sum_{j\geq 1} \big(W_j^* [d_X^E, W_j]+[W_j^*, d_X^E] W_j\big).
	\end{align}
	It is shown in \secref{sec:mult-comm-est} that $\g$ is bounded on $\cH$: 
	\begin{align} \label{kappa} 
		\kappa :=\norm{\g} <\infty,
	\end{align}
	provided assumptions \ref{H1} and \ref{W2} hold. 
	Notice  that the bound on $\kappa$ is independent of $X$, see \eqref{kappa-est}.
	Formally, the velocity operator \eqref{gamma} has a simple origin: 
	\begin{align}\label{gamma'}\g\equiv \g(X,E)=L' (d_X^E),\end{align}
	where  $L'$ is the operator acting on the space of observables $\mathcal{B}(\mathcal{H})$, which is dual to the operator $L$ defined by the r.h.s. of \eqref{vNLeq}, see \eqref{1.2.1} below.

	Under a different set of assumptions, an estimate  similar to \eqref{1.2} is shown in \cite{BFLS} with $O(t^{-n})$ remainder for any $n\ge1$. The assumptions made in \cite{BFLS}  exclude in \eqref{vNLeq} the 
	Schr\"odinger operators  \eqref{H}. 
	
	It is straightforward to show that under the conditions \ref{W1},
	\begin{equation}\label{H'}
		\text{			$ V(x)$ in \eqref{H} is $\Delta$-bounded with relative bound strictly less than $1$,}
	\end{equation}
	and for any $\rho_0\in\cD$ (see \eqref{domain_def}), Eq.~\eqref{vNLeq} has a solution in $\cD$. 
	For more detailed discussions,  see Appendix \ref{sec:exist} below and Refs.~\cite[Section 5.5]{Davies}, \cite[Appendix A]{FFFS}, \cite{OS}.
	Note that Condition \eqref{H'} holds e.g. for every $V\in L^2(\Rb^d)+L^\infty(\Rb^d)$ and  is much weaker than \ref{H1}.

	One can show further (see \cite{AlickiLendi,Davies,FFFS,IngardenKossakowski,kossa} and Appendix \ref{sec:exist}) that the operator $L$ defines a completely positive, trace-preserving, strongly continuous semigroup of contractions. In particular, for any initial state $\rho_0\in\cD$, the solution  $\rho_t,\,t\ge0,$ to \eqref{vNLeq} satisfies
	\begin{equation}\label{1.100}
		\rho_{t}\ge 0, \quad \text{ if } \quad  \rho_{0}\ge 0,  \quad  \text{ and } \quad{\tr\rho_{t}=\tr\rho_{0}}.
	\end{equation}

	Finally, we give the explicit expression of the operator $L'$ in \eqref{gamma'} and its domain. 	Let $L$ be the operator defined by the r.h.s.~of \eqref{vNLeq} on its natural domain $\cD$ (see \eqref{domain_def}), and $L'$ be the operator acting on the space of observables $\mathcal{B}(\mathcal{H})$, which is dual to $L$ with respect to the coupling $(A, \rho):= \tr(A \rho)$, i.e.,
	\begin{equation}\label{1.2.1}
		\tr(A L\rho)= \tr((L'A) \rho), 
	\end{equation}for $\rho\in\mathcal{D}(L)$ and $A\in \mathcal{D}(L')\subset\mathcal{B}(\mathcal{H})$.
	\footnote{$L'$ generates  the dual Heisenberg-Lindblad  evolution $ \partial_t A_t= L' A_t$ of quantum observables.} 
	Explicitly,  the dual vNL operator $L' $ defined in \eqref{1.2.1} is given by:
	\begin{align}\label{L'}
		& L' =L'_0+G',\  \qquad L'_0A 
		=i [H,A],\\
		\label{G'}& G'A:=\frac{1}{2}\sum_{j\geq 1}(W_j^* [A, W_j]+[W_j^*, A] W_j),
	\end{align}
	with domain 
	\begin{align}\label{domain-L'}
		&\mathcal{D}(L')\equiv\,\mathcal{D}(L'_0)\equiv\big\{A\in \mathcal{B}(\mathcal{H})\, | \, A\mathcal{D}(H)\subset\mathcal{D}(H) \text{ and }\notag\\
		&\sbr{H,A}\text{ defined on }\mathcal{D}(A)\cap \cD(H) \text{ extends to an operator on } \mathcal{D}(\mathcal{H}) \big \}.
	\end{align}

	\paragraph{\bf Notation.} 	In the remainder of this paper, $\|\cdot\|$ stands either for the norm of vectors in $\cH$, or for the norm of operators on $\cH$, which one is meant is always clear from the context.	For two bounded operators $A,\,B$, the notation \begin{equation}\label{Oh}A=O(B)\end{equation} means that $\norm{A}\le C_{n,E} \norm{B}$ for some $C_{n,E}>0$ independent of $A\,,B\,,t\,,s$. As above, we will write $$X_a:=\Set{x\in\Rb^d:d_X(x)\le a}\text{ for $a\ge0$},\quad X_{ct}^\cp\equiv (X_{ct})^\cp.$$   {In all our estimates, it is understood that, if $n=1$, the sums $\sum_{k=2}^n(\cdots)$ should be dropped. 

	\medskip

	\section{Recursive monotonicity estimate}\label{sec:RME}

	We work in this section in an abstract setting, with $H$ a self-adjoint operator on a Hilbert space $\cH$ and, for $j=1,2,\dots$, $W_j$ bounded operators in $\cH$ such that $\sum_{j\ge1}W_j^*W_j$ strongly converges in $\cH$. We consider the vNL operator
	\begin{align*}
		&L\rho=-i[H,\rho]+\frac12\sum_{j\geq 1}\big([W_{j},\rho W_{j}^{*}]+[W_{j}\rho,W_{j}^{*}]\big),
	\end{align*}
	defined on the domain \eqref{domain_def}, as well as the dual operator $L'$ defined as in \eqref{1.2.1}--\eqref{domain-L'}.
	
	We consider in addition a self-adjoint operator $\Phi$ on $\cH$, semi-bounded from below. We assume that 
	\begin{equation}\label{Phi-dom-cond-0}
		(\Phi+c)^{-1}\cD(H)\subset \cD(H),
	\end{equation}
	for some $c\ge0$ 
	and there is an integer $n\ge1$ such that, for all $k=1,\dots,n+1$,
	\begin{equation}\label{rme-cond}
		M_k:=1+\norm{\ad{k}{\Phi}{H}}^2 +\|\sum_{j\geq 1}W_j^*W_j\|+\sum_{j\ge 1}\norm{\ad{k}{\Phi}{W_j}}^2<\infty.
	\end{equation}
	Hence 
	\begin{equation}\label{eq:mu-n}
		\mu_n:=\max_{2\le k\le n+1}M_k
	\end{equation}
	is finite.

	Later on, $H$ will be the Schr\"odinger operator \eqref{H} satisfying \ref{H1}, $W_j$ will be bounded operators satisfying \ref{W1}--\ref{W2} and $\Phi$  will be taken to be the operator  $\Phi\equiv \phi^E=g \phi g$ with $g\equiv g(H)$ described in \eqref{g-cond} and some $\phi\in C^\infty(\Rb^d)$,   see Section \ref{sec:mult-comm-est}.
	
	As in \eqref{gamma}--\eqref{kappa} we set
	\begin{align} \label{kappa-0} 
		\kappa_\Phi := \norm{ i[H,\Phi]+\frac{1}{2}\sum_{j\geq 1} \big(W_j^* [\Phi, W_j]+[W_j^*, \Phi] W_j\big) }.
	\end{align}

	The main result of this section is a key differential inequality, \eqref{rme}. The proof of this inequality is \emph{the only place} where the information about equation \eqref{vNLeq} is used. 
	
	\medskip

	\subsection{ASTLO and RME} 
	We construct a class of observables, which we call  \textit{adiabatic spacetime localization observables (ASTLOs)}, which play the central role in our analysis.

	For a constant $\delta>0$ specified later on, we define a set of smooth cutoff functions
	\begin{equation}\label{F}
		\begin{aligned}
			\cX\equiv \cX_{\delta}
			:=&\Set{\chi\in C^\infty(\R)\left|
				\begin{aligned}
					&\supp \chi\subset \Rb_{\ge0}, \supp \chi'\subset (0,\delta/2)\\
					&\chi^\prime\ge 0,\,\sqrt{\chi'}\in C^\infty(\R)
				\end{aligned}\right.
			}.
		\end{aligned}
	\end{equation}
	See \figref{fig:chi} below. 
	\begin{figure}[h]
		\centering
		\begin{tikzpicture}[scale=3]
			\draw [->] (-1,0)--(1.5,0);
			\node [right] at (1.5,0) {$\mu$};
			
			\node [below] at (0,0) {$0$};
			\draw [fill] (0,0) circle [radius=0.02];
			\node [below] at (.65,0) {$\delta/2$};
			\draw [fill] (.65,0) circle [radius=0.02];
			

			\draw [very thick] (-1,0)--(.1,0) [out=5, in=-175] to (.65,1)--(1.5,1);

			\draw [->] (-.1,.5)--(.3,.5);
			\node [left] at (-.1,.5) {$\chi(\mu)$};
			
			
		\end{tikzpicture}
		\caption{Schematic diagram illustrating $\chi\in \cX$.}
		\label{fig:chi}
	\end{figure}
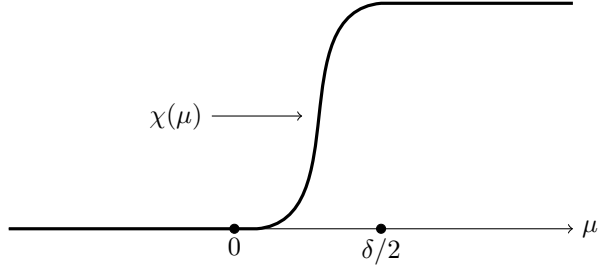
	
	We note that $\chi\ge0$ for $\chi\in\cX$, and the following two properties hold:
	\begin{enumerate}[label=(X\arabic*)]
		\item \label{X1}If $w\in  C_c^\infty$ and  $\supp w\subset (0,\delta/2)$, then the antiderivative $\int^x w^2\in \cX$. 
		\item \label{X2}If $\xi_1,\ldots,\xi_N\in\cX$, then $\xi=(\xi_1^{\frac12}+\cdots+\xi_N^{\frac12})^2$ satisfies $\xi \in \cX$ and $\xi_1+\cdots +\xi_N\le  {\sqrt{N} \xi}$. 
	\end{enumerate}

	For a function $\chi\in \cX$,  a densely defined self-adjoint operator $\Phi$,  a constant $v\in (\kappa,c)$ and $s>t\ge0$, we define a family of self-adjoint operators 
\begin{equation}\label{chi-ts}
	\chi_{ts}=\chi\del{\frac{\Phi-vt}{s}}.
\end{equation} 
Following \cite{BFLS}, we use the method of propagation observables. 
Let $\beta_t'$ be the evolution generated by the operator $L'$, i.e. ${d\over{dt}}\beta'_t(\Psi)=\beta'_t(L'\Psi)$ for all observables $\Psi$ in $\mathcal{D}(L')\subset\cB(\cH)$. 
For a differentiable family of bounded operators $\Psi_t \in \mathcal{D}(L')$, $t\ge0$, we then have the relation  
\begin{align}
	{d\over{dt}} \beta_t'(\Psi_t) =&\beta_t'(D\Psi_t),\label{2.1}\\
	D \Psi_t=&L' \Psi_t 
	+\partial_t\Psi_t.\label{2.2}
\end{align}

As in \cite{BFLS}, we call the operation $D$ the {\it Heisenberg derivative}.

Note that the condition \eqref{Phi-dom-cond-0} ensures that for all $t$, $s$, the bounded observable $\chi_{ts}$ belongs to the domain of $L'$ and also that the commutator expansion \lemref{lemA.2} can be applied. The main result of this section is the following:

\begin{theorem}[recursive monotonicity estimate]\label{thm5.1}
	Suppose that \eqref{Phi-dom-cond-0}--\eqref{rme-cond} hold. Let $\chi\in\cX$ and let $\chi_{ts}$ be the operator defined in \eqref{chi-ts}. 
	Then there exists $C=C(n,\chi)>0$ and, if $n\ge2$, $\xi^k=\xi^k(\chi)\in\cX,\,k=2,\ldots,n$, such that as self-adjoint operators,
	\begin{align}\label{rme}
		D\chi_{ts}\le -\frac{v-\kappa_\Phi}{s}\chi'_{ts}+ \sum_{k=2}^n\frac{M_k}{s^{k}} (\xi^k)'_{ts}+C\frac{\mu_{n}}{s^{n+1}},
	\end{align}
	where $\kappa_\Phi>0$ is as in \eqref{kappa-0} and $M_k$ and $\mu_n$ are defined in \eqref{rme-cond} and \eqref{eq:mu-n}.
\end{theorem}
This theorem is proved in \secref{sec:4}.

Since the second, remainder term on the r.h.s. is of the same form as the leading, negative term, we call \eqref{rme} the \textit{recursive monotonicity estimate (RME)}. It can be bootstrapped as in \propref{prop:propag-est1} to obtain an integral inequality with $O(s^{-n})$ remainder. We write, for $r\ge0$,
\begin{align}\label{2.2'-0}
	\chi_{ts}(r):= \beta_r'(\chi_{ts})\ \qquad \text{ and }\ \quad  \chi_{ts}'(r):= \beta_r'(\chi_{ts}').
\end{align}

\begin{proposition}\label{prop:propag-est1} 
	Suppose the assumptions  of \thmref{thm5.1} hold.
	Then, for all $c>\kappa_{\Phi}$ and $\chi\in \cX$, 
	there exist $C=C(n,\chi)>0$ and  $\xi^k\in \cX$, $2\le k\le n$ (dropped for $n=1$),  such that for all   
	$0\le t< s$, 
	\begin{align}
		&\int_0^t \chi_{rs}'(r) dr  \le C\mu_n^{n} \Big(s\chi_{0s}(0) 
		+ \sum_{k=2}^n s^{-k+2} \ \xi^k_{s}(0) +   {t}s^{-n}\Big) ,
		\label{propag-est31} 
	\end{align}
	where $\mu_n$ is given by \eqref{eq:mu-n}.
\end{proposition} 
\begin{remark}
Instead of the evolution $\chi_{rs}(t)$, we could have used the expectation:
\begin{equation}\label{br}
	\lan \chi_{ts}\ran_t :=\Tr(\chi_{ts}\rho_t)
\end{equation}
of $\chi_{ts}$ in the state $\rho_t$ solving \eqref{vNLeq} and instead of \eqref{2.1}, used the relation 
\begin{align}
	{d\over{dt}}\br{\chi_{ts}}_t =&\br{D\chi_{st}}_t.\label{2.1'}
\end{align}
These two formulations are related as
\begin{align}\label{2form-rel}			\left<\chi_{ts}\right>_t =\left<\chi_{ts}(t)\right>_0.
\end{align}	
\end{remark}

\medskip

\subsection{Proof of Theorem \ref{thm5.1}}\label{sec:4}

To prove the recursive monotonicity estimate, \thmref{thm5.1}, we first need a totally symmetrized commutator expansion. Our next results, \propref{prop:H-contri} and \propref{prop:W-contri}, generalize   the commutator expansion for bounded operators,   first obtained in  \cite{SigSof}, and subsequently improved in e.g. \cite{GoJe,HunSig1, HunSigSof, Skib}.  We refer to \cite{HunSig1} for details and references.

Recall that the dual vNL operator $L'$ satisfies $L'=i[H,A]+G'A$ for all $A$ in $\cD(L')$, where $G'$ is given by \eqref{G'}.

\begin{proposition}
\label{prop:H-contri} 	
Suppose that \eqref{Phi-dom-cond-0} and \eqref{rme-cond} hold.
Let  $\chi\in\cX$ and let $\chi_{ts}$ be the operator defined by \eqref{chi-ts}. 		
Then, uniformly in $t$, for $s>0$,
\begin{equation}
	\label{H-esti}
	i[H,\chi_{ts}]=s^{-1}\sqrt{\chi_{ts}'}i[H,\Phi]\sqrt{\chi_{ts}'}+\Rem_H
\end{equation}
where the remainder term $\Rem_H$ satisfies the estimate
\begin{align}
	\pm \Rem_{H}& \leq \sum_{k=2}^{n}\frac{M_k}{s^k} (\xi^k)'_{ts}+C\frac{M_{n+1}}{s^{n+1}}
\end{align}
for some $\xi^2,...,\xi^n\in\cX$ depending only on $\chi$, with $M_k$ as in \eqref{rme-cond} and for some constant $C=C(n,\chi)>0$.
\end{proposition}

\begin{proposition}\label{prop:W-contri}
Suppose that \eqref{Phi-dom-cond-0} and \eqref{rme-cond} hold.
Let  $\chi\in\cX$ and let $\chi_{ts}$ be the operator defined by \eqref{chi-ts}. 		
Then, uniformly in $t$, for $s>0$,
\begin{align}\label{W-esti}
	G'(\chi_{ts})&=s^{-1}\sqrt{\chi'_{ts}}G'(\Phi)\sqrt{\chi_{ts}'}+\Rem_W,
\end{align}
where the remainder term $\Rem_W$ satisfies the estimate
\begin{align}
	\pm \Rem_W&\leq \sum_{k=2}^{n}\frac{M_k}{s^k}(\xi^k)'_{ts}+C\frac{\mu_n}{s^{n+1}}
\end{align}
for some $\xi^2,...,\xi^n\in\cX$ depending only on $\chi$, for some constant $C=C(n,\chi)>0$, with $M_k$ and $\mu_{n}$ as in \eqref{rme-cond}  and \eqref{eq:mu-n}.
\end{proposition}

\begin{remark}
The estimates above are all uniform in $s,t,\Phi$ and, in particular, are valid for the operator $\phi^E=g \phi g$ such as \eqref{5.2.2}. 
\end{remark}				

\begin{remark}
We note that the error term in \thmref{thm1} arises in the symmetrization procedure above, and can be improved as the expansion continues to higher order.
\end{remark}

\begin{proof}[Proof of Proposition \ref{prop:H-contri}]
In this proof, the time $t$ is fixed and is omitted from the notation, so we write $\chi_s$ for $\chi_{ts}$.  Also, we denote $B_k\equiv i\ad{k}{\Phi}{H}$ for $k=1,...,n+1$.  In this case, since $H$ is self-adjoint, we have $B_k^{*}=(-1)^{k-1}B_k$.

1. By \eqref{Phi-dom-cond-0}--\eqref{rme-cond} and the assumption on $\chi$, the hypotheses of Lemma \ref{lemA.2} are satisfied.  Hence, by \eqref{4.Hcomm-exp}--\eqref{4.Hcomm-exp-right}, we have
\begin{align}\label{A-sym-exp}
	i[H,\chi_s]&=\frac{1}{2}\sum_{k=1}^{n}\frac{s^{-k}}{k!}\left(\chi_s^{(k)}B_k+B_k^*\chi_{s}^{(k)}\right)+\frac{1}{2}s^{-(n+1)}\left(R_{n+1}+R_{n+1}^*\right),
\end{align}
where $\|R_{n+1}\|\leq c\|B_{n+1}\|$ for some constant $c>0$ depending only on $\chi$.

\medskip

2. Next, we claim that every term on the r.h.s. of \eqref{A-sym-exp}, except for the leading term $(k=1)$, are uniformly bounded by $(\chi_1)'_{s}$ for some $\chi_1\in\cX$.  

To show this, for each $k$, we choose some smooth function $\theta^k\in C_{c}^{\infty}((0,\delta/2))$ that takes value $1$ on $\supp(\chi^{(k)})$.  Then, we claim that 
\begin{align}
	\chi_{s}^{(k)}B_k=\chi_{s}^{(k)}B_{k}\theta_{s}^{k}+O(s^{-(n+1-k)}),
\end{align}
where $\theta_s^k\equiv \theta^k(s^{-1}(\Phi-vt))$.  Indeed, using commutator expansion and the fact that $\Ad_{\Phi}^l(B_k)=B_{k+l}$, we have
\begin{align}
	\chi_{s}^{(k)}B_k&=\chi_{s}^{(k)}\theta_{s}^{k}B_k=\chi_{s}^{(k)}B_{k}\theta_{s}^{k}+\chi_{s}^{(k)}[\theta_s^{k},B_k]\nonumber\\
	&=\chi_{s}^{(k)}B_k\theta_{s}^{k}-\chi_{s}^{(k)}\sum_{l=1}^{n-k}\frac{(-1)^ls^{-l}}{l!}(\theta^k)_{s}^{(l)}B_{k+l}\nonumber\\
	&\quad\quad\quad\quad+(-1)^{n+1-k}s^{-(n+1-k)}\chi_{s}^{(k)}\Rem_{\rm right}(s),
\end{align}
where
\begin{align}
	\Rem_{\rm right}(s)&=\int d\widetilde{\theta^k}(z)R^{n+1-k}B_{n+1}R.
\end{align}
Since $\theta^k$ has compact support, $\Rem_{\rm right}(s)$ is bounded so that 
\begin{align}
	\chi_{s}^{(k)}B_k&=\chi_{s}^{(k)}B_k\theta_s^k-\chi_s^{(k)}\sum_{l=1}^{n-k}\frac{(-1)^ls^{-l}}{l!}(\theta^k)_s^{(l)}B_{k+l}+O(s^{-(n+1-k)}).
\end{align}
Next, since $\theta^k\equiv 1$ on $\supp(\chi^{(k)})$, we have $\supp((\theta^k)^{(l)})\cap\supp(\chi^{(k)})=\varnothing$ for all $l\geq 1$ so that 
\begin{align}
	\chi_{s}^{(k)}\sum_{l=1}^{n-k}\frac{(-1)^ls^{-l}}{l!}(\theta^k)_{s}^{(l)}B_{k+l}=0.
\end{align}
It follows that 
\begin{align*}
	\chi_{s}^{(k)}B_{k}&=\chi_s^{(k)}B_k\theta_{s}^k+O(s^{-(n+1-k)})
\end{align*}
so that 
\begin{align}\label{new-opera-esti}
	s^{-k}(\chi_{s}^{k}B_k+B_{k}^*\chi_{s}^{k})&=s^{-k}(\chi_{s}^{k}B_k\theta_s^{k}+\theta_{s}^{k}B_k^*\chi_{s}^{k})+O(s^{-(n+1)}).
\end{align}

Now, we apply the following operator inequality
\begin{align}\label{op-esti}
	\pm (P^*Q+Q^*P)&\leq P^*P+Q^*Q.
\end{align}
with $P=\chi_{s}^{(k)}$ and $Q=B_k\theta_s^k$ on \eqref{new-opera-esti} to obtain
\begin{align}\label{opera-ineq-Bk}
	s^{-k}(\chi_{s}^{k}B_k+B_{k}^*\chi_{s}^{k})&\leq s^{-k}\left((\chi_s^{(k)})^2+\|B_k\|^2(\theta_{s}^k)^2\right)+O(s^{-(n+1)}).
\end{align}
Since $n$ is finite, we can choose $\xi^2,...,\xi^n\in \cX$ such that $(\xi^k)'$ majorizes $(\chi^{(k)})_s^2+\|B_k\|^2(\theta^k_s)^2$ for each $k$.

\medskip

3. Now, we symmetrize the leading order term.  Let $u=(\chi')^{1/2}$.  Since $u$ is smooth by assumption, we use \eqref{HSj-rep} to expand the leading order terms and obtain
\begin{align}
	(u_s)^2B_1+B_1(u_s)^2&=2u_sB_1u_s+u_s[u_s,B_1]+[B_1,u_s]u_s\nonumber\\
	&=2u_sB_1u_s+\sum_{l=1}^{n-1}\frac{s^{-l}}{l!}\left(u_s u_s^{(l)}B_{1+l}+B_{1+l}^{*}u_{s}^{(l)}u_{s}\right)\nonumber\\
	&\quad\quad+s^{-n}(u_sR_{n}'+R_n'^{*}u_s),
\end{align}
where $\|R_n'\|\leq c'\|B_{n+1}\|$ for some constant $c'>0$ depending only on $u$.  

Again, using operator estimate \eqref{op-esti}, for each $l=1,...,n-1$, we have
\begin{align}
	s^{-l}(u_s u_s^{(l)}B_{1+l}+B_{1+l}^{*}u_s^{(l)}u_s)&\leq s^{-1}\|B_{1+l}\|^2(u_s^{(l)})^2+s^{-2l+1}(u_s)^2,
\end{align}
and for the remainder term we have
\begin{align}
	s^{-n}(u_sR_n'+R_n'^{*}u_s)&\leq s^{-1}(u_s)^{2}+s^{-2n+1}\|R_n'\|^2(\tilde{\theta}_s)^2,
\end{align}
where $\tilde{\theta}$ is again some smooth cutoff function supported in $(0,\delta/2)$ that takes value $1$ on the support of $u$ and $\tilde{\theta}_s\equiv \tilde{\theta}(s^{-1}(\Phi-vt))$.  Since $u$, $u^{(l)}$ and $\tilde{\theta}$ are supported in $(0,\delta/2)$, we can modify $\xi^2,...,\xi^n$ in such a way that $\xi^l\in\cX$ majorizes $u^2$, $\tilde{\theta}^2$ and $(u^{(l)})^{2}$ for each $l=1,...,n-1$.  

Collecting all terms except for the leading order ones into the remainder term $\Rem_{H}$, we obtain \eqref{H-esti}.
\end{proof}	
\begin{proof}[Proof of Proposition \ref{prop:W-contri}]
In this proof, we also fix $t$ and omit it from the notation.  Furthermore, we fix $j\geq 1$ and denote $D_{j,k}\equiv \ad{k}{\Phi}{W_j}$.  In particular, we obtain $\ad{k}{\Phi}{W_j^*}=(-1)^k(\ad{k}{\Phi}{W_j})^*=(-1)^{k}D_{j,k}^*$.

\medskip

1. First, using Lemma~\ref{lemA.2} 
and the boundedness of $W_j$, we have
\begin{align}\label{Wj-exp}
	[\chi_{s},W_j]&=-\sum_{k=1}^{n}\frac{s^{-k}}{k!}\chi_s^{(k)}D_{j,k}-s^{-(n+1)}R^{\rm right}_{j,n+1}
\end{align}
where $R^{\rm right}_{j,n+1}$ is given in \eqref{right-rem} and satisfies the estimate
\begin{align}
	\|R^{\rm right}_{j,n+1}\|^2&\leq C\|D_{j,n+1}\|^2, \label{R<D}
\end{align}
for some constant $C$ independent of $j$.  Similarly, we have
\begin{align}\label{Wj*-exp}
	[W_j^*,\chi_s]&=-\sum_{k=1}^{n}\frac{s^{-k}}{k!}D_{j,k}^*\chi^{(k)}_{s}-(-1)^{n+1}s^{-(n+1)}\widetilde{R}^{\rm left}_{j,n+1},
\end{align}
where $\widetilde{R}^{\rm left}_{j,n+1}=(-1)^{n+1}(R^{\rm right}_{j,n+1})^{*}$.  Combining \eqref{Wj-exp} and \eqref{Wj*-exp}, we have
\begin{align}\label{Gj}
	G_{j}'(\chi_{s})&=W_{j}^*[\chi_s,W_j]+[W_j^*,\chi_s]W_j\nonumber\\
	&=-\sum_{k=1}^{n}\frac{s^{-k}}{k!}\left(W_j^*\chi_{s}^{(k)}D_{j,k}+D_{j,k}^{*}\chi_{s}^{(k)}W_j\right)\nonumber\\
	&\quad\quad -s^{-(n+1)}\left(W_j^*R^{\rm right}_{j,n+1}+(R^{\rm right}_{j,n+1})^*W_j\right),
\end{align}
where $G_j'(\cdot)=W_j^*[\cdot, W_j]+[W_j^*,\cdot]W_j$.

\medskip 

2. We now verify that the r.h.s. of \eqref{Gj} is summable in $j\geq 1$.  We begin with the remainder terms.  Using the operator estimate \eqref{op-esti}, we obtain
\begin{align}
	\pm \left(W_j^*R^{\rm left}_{j,n+1}+(R^{\rm left}_{j,n+1})^*W_j\right)&\leq W_j^*W_j+\|R^{\rm left}_{j,n+1}\|^2,
\end{align}
which are summable in $j\geq 1$ since $\sum_jW_j^*W_j$ strongly converges in $\cH$, and since \eqref{rme-cond} and \eqref{R<D} hold.

Next, we estimate the $k$-th terms in the first two lines of \eqref{Gj}.  Let $\theta^k$ be some smooth cutoff function supported in $(0,\delta/2)$ such that $\theta^k\equiv 1$ on $\supp(\chi^{(k)})$.  It follows that $\chi_{s}^{(k)}=\theta^k_s\chi_{s}^{(k)}\theta^k_s$, where $\theta^k_s\equiv \theta^k(s^{-1}(\Phi-vt))$.  Then, we claim that 
\begin{align}
	W_j^*&\chi_s^{(k)}D_{j,k}+D_{j,k}^*\chi_s^{(k)}W_j\nonumber\\
	&=\theta^k_s \left(W_j^*\chi_s^{(k)}D_{j,k}+D_{j,k}^*\chi_s^{(k)}W_j\right) \theta^k_s+s^{-(n+1-k)}C_k\|D_{j,n+1}\|^2\label{k-esti},
\end{align}
where $C_k$ is some constant depending only on $\chi^{(k)}$.  

If \eqref{k-esti} holds, then using \eqref{op-esti}, we have
\begin{align}
	\pm &\left(W_j^*\chi_s^{(k)}D_{j,k}+D_{j,k}^*\chi_s^{(k)}W_j\right)\nonumber\\
	&\quad\quad\leq \theta^k_sW_j^*W_j\theta^k_s+\|D_{j,k}\|^2\|\chi^{(k)}\|^2(\theta^k_s)^2+C_ks^{-(n+1-k)}\|D_{j,n+1}\|^2,
\end{align}
which are also summable in $j\geq 1$ by \eqref{rme-cond}.

\medskip 

3. Now, we prove the claim \eqref{k-esti}.  By a direct calculation, we have
\begin{align}
	W_j^*\chi_{s}^{(k)}D_{j,k}&-\theta^k_s W_j^*\chi_{s}^{(k)}D_{j,k} \theta^k_s \nonumber\\
	&=[W_j^*,\theta^k_s]\chi_{s}^{(k)}D_{j,k}+\theta^k_sW_j^*\chi_{s}^{(k)}[\theta^k_s,D_{j,k}]
\end{align}
and a similar expression for $D_{j,k}^*\chi_{s}^{(k)}W_j$.  Thus, it suffices to show that $[W_j^*,\theta^k_s]$ and $[\theta^k_s,D_{j,k}]$ are $O(s^{-(n-k)})$.  

\medskip

3.1. For the first term, we use \eqref{Wj*-exp} to obtain
\begin{align}
	[W_j^*,\theta^k_s]&=-\sum_{l=1}^{n}\frac{s^{-l}}{l!}D_{j,l}^{*}(\theta^k)^{(l)}_s-s^{-(n+1)}R_{j,n+1}^*,
\end{align}
where $R_{j,n+1}$ is given by \eqref{right-rem} and satisfies the estimate $\|R_{j,n+1}\|\leq C\|D_{j,n+1}\|$.  Since $\theta^k\equiv 1$ on $\supp(\chi^{(k)})$, then we have $(\theta^k)^{(l)}_s\chi_s^{(k)}=0$ for $l\geq 1$ so that 
\begin{align}
	[W_j^*,\theta^k_s]\chi_{s}^{(k)}D_{j,k}&=-s^{(n+1)}R_{j,n+1}^*\chi^{(k)}_s D_{j,k},
\end{align}
which is $O(s^{-(n+1)})$ and summable in $j\geq 1$, by the Cauchy-Schwarz inequality and \eqref{rme-cond}.

\medskip 

3.2. For the second term, we proceed similarly, using \eqref{Wj-exp}, to obtain
\begin{align}
	[\theta^k_s,D_{j,k}]&=-\sum_{l=1}^{n-k}\frac{s^{-l}}{l!}(\theta^k)^{(l)}_sD_{j,k+l}+s^{-(n+1-k)} \widetilde{R}_{j,n+1-k},
\end{align}
where $\widetilde{R}_{j,n+1-k}$ is given by \eqref{left-rem} with $n$ replaced by $n-k$ and satisfies the estimate $\|\widetilde{R}_{j,n+1-k}\|\leq C\|D_{j,n+1-k}\|$ with $C$ only depending on $\theta^k$.  Using the same reason as above, since $\chi_{s}^{(k)}(\theta^k)_s^{(l)}=0$ for all $l\geq 1$, we conclude that
\begin{align}
	\theta^k_sW_j^*\chi_s^{(k)}[\theta^k_s,D_{j,k}]&=s^{-(n+1-k)}\theta^k_sW_j^*\chi_s^{(k)}\widetilde{R}_{j,n+1-k}.
\end{align}
This completes the proof of the claim \eqref{k-esti}.

\medskip

4. Now we choose $\xi^2,...,\xi^n\in\cX$ such that 
\begin{equation*}
	\left(\|\sum_{j\geq 1}W_j^*W_j\|+\sum_{j\geq 1}\|D_{j,k}\|^2\right)(\theta^k)^2\le M_k (\xi^k)'.
\end{equation*}
Then, by writing everything as $\Rem_W$ in \eqref{Gj} except for the leading order terms (obtained for $k=1$), we obtain, up to some terms coming from the leading order terms which will be dealt with below, the estimate
\begin{align}
	\pm \Rem_W&\leq \sum_{k=2}^{n+1}\frac{M_k}{s^k}(\xi^k)'_s+\frac{C\mu_{n}}{s^{n+1}},
\end{align}
where $C$ is a constant depending only on $\chi$ and $n$. 

\medskip 

5. Finally, we deal with the leading order terms (obtained for $k=1$) in \eqref{Gj}.  Following the same lines as in the proof for Proposition \ref{prop:H-contri}, we define $u=\sqrt{\chi'}$ and use \eqref{HSj-rep} to obtain
\begin{align}
	&W_j^*\chi_{s}'D_{j,1}+\text{h.c}.\notag\\=&u_sW_j^*D_{j,1}u_s+[W_j^*,u_s]u_sD_{j,1}+u_sW_j^*[u_s,D_{j,1}]+\text{h.c}.,
\end{align}
where $\text{h.c}.$ means the adjoint of the terms before it.  Without repeating the same calculation as above, using \eqref{HSj-rep} and \eqref{op-esti}, we can show that the commutators are summable in $j\geq 1$.  Then, we modify $\xi^k\in\cX$ to majorize $(u^{(k)})^2$ and $u^2$ as well.  This completes the proof.
\end{proof}


Now we are ready to prove Theorem \ref{thm5.1}:

\begin{proof}[Proof of Theorem \ref{thm5.1}]
Given Proposition \ref{prop:H-contri}--\ref{prop:W-contri}, we choose $\xi^2,...,\xi^n$ depending on $\chi$, in such a way that 
\begin{align}
	\Rem_H+\Rem_W&\leq \sum_{k=2}^{n}\frac{M_k}{s^{k}} (\xi^k)'_{ts}+C\frac{\mu_{n}}{s^{n+1}},
\end{align}
where 
$C$ is some constant which depends only on $n$ and $\chi$.  

It remains to calculate $\partial_t\chi_{ts}$.  Using the chain rule, we immediately obtain
\begin{align}\label{eq:deriv}
	\partial_t\chi_{ts}&=-s^{-1}v\chi_{ts}'.
\end{align}
This completes the proof.
\end{proof}

\medskip

\subsection{Proof of \propref{prop:propag-est1}}

\begin{proof}[Proof of \propref{prop:propag-est1}]
Within this proof, 	all constants $C>0$ depend only on $\chi$ and $n$.

We will use the relation \eqref{2.1}. First, we observe that, 
by Condition \eqref{Phi-dom-cond-0} and Definition \eqref{F}, for $\chi\in\cX$ and all  $0< t\le s$, the operator $\chi_{ts}$ maps $\cD(H)$ into itself. Moreover, \eqref{A-sym-exp} in the proof of Proposition \ref{prop:H-contri} shows that $[H,\chi_{ts}] \in \cB(\cH)$. Hence $\chi_{ts}\in\mathcal{D}(L')$. 

Next, for each fixed $s$, integrating the formula \eqref{2.1} with $\Psi_t\equiv \chi_{ts}$ in $t$ gives
\begin{align} \label{eq-basic}  
	\chi_{ts}(t)-\int_0^t \beta'_r(D\chi_{rs})\,dr= \chi_{0s}(0).
\end{align}
The positive-preserving property of the flow \eqref{vNLeq} (see \eqref{1.100}) extends by duality to $\beta'_r$\,, so that we can
apply the  inequality \eqref{rme} to the second term on the l.h.s. of \eqref{eq-basic} to obtain 
\begin{align} \label{propag-est2} 
	&\chi_{ts}(t)+(v-\kappa_{\Phi})s^{-1}\int_0^t\chi'_{rs}(r)\, dr\notag\\
	& \le \chi_{0s}(0)+C \mu_n \del{ \sum_{k=2}^ns^{-k}\int_0^t (\xi^k)'_{rs}(r)\,dr + t s^{-(n+1)}},
\end{align}
where we recall that the second term in the r.h.s. is dropped for $n=1$.

Since $\kappa_{\Phi} < v$ and $t\le s$, \eqref{propag-est2} implies, after dropping $\chi_{ts}(t)\ge0$, which is due to the positive-preserving property of the flow \eqref{vNLeq} (see \eqref{1.100}), and multiplying by $s(v-\kappa_{\Phi})^{-1}\ge0$, that
\begin{align}
	\int_0^t \chi'_{rs}(r)\, dr 
	\le C \mu_n
	\del{s\chi_{0s}(0) +   \sum_{k=2}^ns^{-k+1}\int_0^t (\xi^k)'_{rs}(r)\,dr+ t s^{-n}}. \label{propag-est3} 
\end{align}

3. If $n=1$, then \eqref{propag-est3} gives \eqref{propag-est31}. If $n\ge2$, applying \eqref{propag-est3} to the term $\int_0^t (\xi^{k})'_{rs}(r)\,dr$ and using the property \ref{X2}, we obtain
\begin{align}
	\int_0^t\chi'_{rs}(r)\, dr \, 
	\le C \mu_n^{2} \Big( \, &
	s \chi_{0s}(0)+ \xi^2_{0s}(0) + \sum_{k=3}^n s^{-k+2} \int_0^t (\eta^k)'_{rs}(r)\,dr +  t s^{-n}\Big), \label{propag-est33} 
\end{align}
where the third term in the r.h.s. is dropped for $n=2$, and $\eta^k=\eta^k(\xi^2,\xi^k)\in\cX, k=3,\ldots, n$. Bootstrapping this procedure, we arrive at \eqref{propag-est31}.
\end{proof}

\medskip

\section{Proof of \thmref{thm1}} 
\label{sec:main-thm-pf} 

We formulate the technical relations mentioned in \thmref{thm1}.  Given a smooth, non-negative cutoff functions $g$  with $\supp(g)\subset (-\infty,E]$ (see also \remref{remG})and a smooth function $\chi$ from the space \eqref{F}, we choose smooth cutoff functions $\tilde{g}$ and $\tilde \chi$ such that $\supp(\tilde{g})\subset \{g\equiv 1\}$ and  $\supp(\tilde{\chi}')\subset (\delta,+\infty)=\{\chi\equiv1\}$, so that
\begin{align}
	&\bar{\chi}(\mu)\tilde{\chi}(\mu)=0,\label{tildechi-cond'}\\
	&\bar{g}(\mu)\tilde{g}(\mu)=0.\label{tildeg-cond'}
\end{align}
see Figs.~\ref{fig:tildechi}--\ref{fig:tildeg}.

	\begin{figure}[t]
	\centering
	\begin{tikzpicture}[scale=2.5]
		\draw [->] (-1,0)--(2.15,0);
		\node [right] at (2.15,0) {$\mu$};
		\node [below] at (0,0) {$0$};
		\draw [fill] (0,0) circle [radius=0.02];
		\node [below] at (.65,0) {$\delta/2$};
		\draw [fill] (.65,0) circle [radius=0.02];
		\node [below] at (1.3,0) {$\delta$};
		\draw [fill] (1.3,0) circle [radius=0.02];
		\draw [very thick] (-1,0)--(.1,0) [out=5, in=-175] to (.65,1)--(2.15,1);
		\draw [very thick,dashed] (-1,0)--(.75,0) [out=5, in=-175] to (1.3,1)--(2.15,1);						
		\draw [->] (-.1,.5)--(.3,.5);
		\node [left] at (-.1,.5) {$\chi(\mu)$};
		\draw [->] (1.5, .5)--(1.1,.5);
		\node [right] at (1.5,.5) {$\tilde \chi(\mu)$};						
	\end{tikzpicture}
	\caption{Schematic diagram illustrating $\tilde \chi$ satisfying \eqref{tildechi-cond'}.}\label{fig:tildechi}
\end{figure}
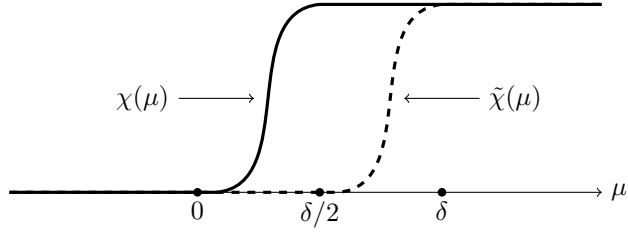

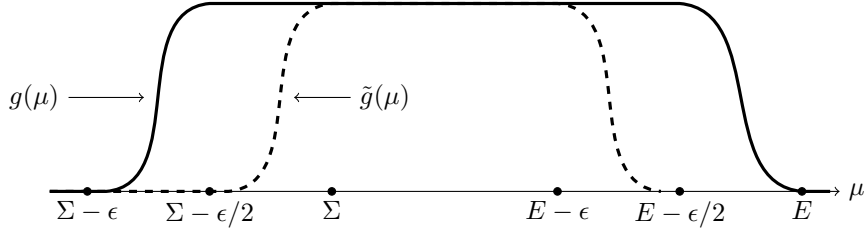
\begin{figure}[t]
	\centering
	\begin{tikzpicture}[scale=2.5]
		\draw [->] (-.2,0)--(4,0);
		\node [right] at (4,0) {$\mu$};
		\node [below] at (0,0) {$\Si-\eps$};
		\draw [fill] (0,0) circle [radius=0.02];
		\node [below] at (.65,0) {$\Si-\eps/2$};
		\draw [fill] (.65,0) circle [radius=0.02];
		\node [below] at (1.3,0) {$\Si$};
		\draw [fill] (2.5,0) circle [radius=0.02];
		\node [below] at (2.5,0) {$E-\eps $};
		\draw [fill] (3.15,0) circle [radius=0.02];
		\node [below] at (3.15,0) {$E-\eps/2 $};
		\draw [fill] (3.8,0) circle [radius=0.02];
		\node [below] at (3.8,0) {$E $};
		\draw [fill] (1.3,0) circle [radius=0.02];
		\draw [very thick] (-.2,0)--(.1,0) [out=5, in=-175] to (.65,1)--(3.15,1) [out=-5, in= 175] to (3.8,0)--(3.95,0) ;
		\draw [very thick,dashed] (-.2,0)--(.75,0) [out=5, in=-175] to (1.3,1)--(2.5,1)[out=-5, in= 175] to(3.05,0);						
		\draw [->] (-.1,.5)--(.3,.5);
		\node [left] at (-.1,.5) {$g(\mu)$};
		\draw [->] (1.4, .5)--(1.1,.5);
		\node [right] at (1.4,.5) {$\tilde g(\mu)$};						
	\end{tikzpicture}
	\caption{Schematic diagram illustrating $\tilde g$ satisfying \eqref{tildeg-cond'}. Here $\Si:=\inf\si(H)$ (see \remref{remG}).}\label{fig:tildeg}
\end{figure}

We also specify the self-adjoint operator $\Phi$ in \thmref{thm5.1} and definition \eqref{chi-ts} as  
\begin{equation}\label{5.2.2}
\Phi:=d_X^E=g(H)d_Xg(H), 
\end{equation}
where, recall, $X\subset \Rb^d$ is a bounded subset with smooth boundary and  $d_X\in C^\infty(\Rb^d)$ is the smoothed distance function to $X$ given in \eqref{dX} for some $\eps_0>0$ and satisfies \eqref{dX2}.

To shorten notations, we introduce the following notations:
\begin{align}\label{chitsdef}
\chi_{ts}^E:=\chi((d_X^E-vt)/s),\quad \chi_{ts}:=\chi((d_X -vt)/s).
\end{align}
Now, for any $\chi\in\cX$ and $\tilde g,\,\tilde \chi$ as above, we claim that	
\begin{align}
\label{claim1}
&	\chi_X^\sharp\chi_{0s}^E\chi_X^\sharp=O(s^{-n}),\\
&			\chi_{ts}^E\geq \tilde{g}\tilde{\chi}_{ts}\tilde{g}+O(s^{-n}),\label{claim2}
\end{align}
where we recall that $\chi_X^\sharp$ stands for the characteristic function of $X$. We discuss these claims in Section \ref{sec:claims-deriv}.

Recall that $\beta_t'$ denotes the evolution generated by the operator $L'$ and that $\chi_{ts}^E(t):= \beta_t'(\chi_{ts}^E)$, $(\chi')^E_{ts}(t):= \beta_t'((\chi')_{ts}^E)$. 
We are now ready to prove \thmref{thm1}.

\begin{proof}[Proof of \thmref{thm1}]

We want to apply Proposition \ref{prop:propag-est1} to $H=-\Delta+V(x)$ and $W_j$ satisfying \ref{H1}--\ref{W2}, with $\Phi$ given by \eqref{5.2.2}. Hence we need to verify that the abstract conditions \eqref{Phi-dom-cond-0}--\eqref{rme-cond} are satisfied.

First, we fix any $c>0$ and justify that $(d_X^E+c)^{-1}$ maps $\cD(H)$ into itself. Recalling that $d_X^E=g(H)d_Xg(H)$ with $\mathrm{supp}(g)\subset (-\infty,E]$, we have  
\begin{align*}
	(d_X^E+c)^{-1} &= \chi_{(-\infty,E]}^\#(H)
	(d_X^E+c)^{-1}+\chi_{(E,\infty)}^\#(H)(d_X^E+c)^{-1}\\
	&=\chi_{(-\infty,E]}^\#(H)(d_X^E+c)^{-1}+c^{-1}\chi_{(E,\infty)}^\#(H).
\end{align*}
The first term is a bounded operator from $\cH$ to $\cD(H)$ while the second term obviously preserves $\cD(H)$. This shows that $(d_X^E+c)^{-1}$ maps $\cD(H)$ into itself


Next, condition \eqref{rme-cond} is verified in Section \ref{sec:mult-comm-est}, see Corollary \ref{cor:4.3}.  Therefore Proposition \ref{prop:propag-est1} with $\Phi=d_X^E$ applies. 

{Now we take $\chi\in\cX$ with $\chi(\mu)\equiv 1$ for $\mu\ge\delta/2$.}	Retaining the first term in the l.h.s. of \eqref{propag-est2} in the proof of Proposition \ref{prop:propag-est1} and dropping the second one, which is non-negative since $\chi'\ge0$ and $v>\kappa$, we obtain
\begin{align*} 
\chi_{ts}^E(t)  \le \,\chi_{0s}^E(0)+C_{n,E}\del{ \sum_{k=2}^ns^{-k+1}\int_0^t ((\xi^k)')^E_{rs}(r) dr + t s^{-(n+1)}}.		
\end{align*}
Here we used that the constant $\mu_n=\max_{2\le k\le n+1}M_k$ appearing in the r.h.s. of $\eqref{propag-est2}$ is bounded by $C_{n,E}$ for some positive constant depending on $n$ and $E$.
Applying \eqref{propag-est31} 
to the second term on the r.h.s.,

we deduce that, with the notation as in \eqref{Oh},		
\begin{align} \label{propag-est4} 
\chi_{ts}^E(t)  \le 	\chi_{0s}^E(0)+O(s^{-1}\xi_{0s}^E(0))+O(s^{-n}), \end{align}
for some $\xi\in\cX$  and all  $s> t$. Taking expectation w.r.t.~$\rho_0$ on both sides of \eqref{propag-est4} and recalling that $\chi_{ts}(t):= \beta_t'(\chi_{ts})$, we find
\begin{equation}\label{3.40}
\omn{\beta'_t(\chi_{ts}^E)}  \le	\omn{\del{\chi_{0s}^E+O(s^{-1}\xi_{0s}^E)}}+O(s^{-n}).
\end{equation}

By the localization assumption \eqref{lam-cond} on the initial state, we have  $\rho_0=\chi_X^\sharp\rho_0\chi_X^\sharp$. By this fact, 
we find 
\begin{equation}\label{3.41}
\omn{\del{\chi_{0s}^E+O(s^{-1}\xi_{0s}^E)}}= \omn{\chi_X^\sharp \del{\chi_{0s}^E+O(s^{-1}\xi_{0s}^E)}\chi_X^\sharp}= O(s^{-n}).
\end{equation}
The relation \eqref{claim2} implies \begin{align}\label{3.41-0}
\chi_{ts}^E&\geq \tilde{g}\tilde{\chi}_{ts}\tilde{g}+O(s^{-n}),
\end{align}
where we recall that $\tilde{g}$ is a smooth non-negative cutoff function with $\supp(\tilde{g})\subset\{g\equiv 1\}$ and $\tilde{\chi}$ is a smooth function such that $\tilde \chi\equiv 1 $ on $ (\delta,+\infty)$. 
It follows that, by applying the dual evolution $\beta_t'$, 
\begin{equation}
\label{hatchi-ts-est} \beta_t'(\tilde{g}\tilde{\chi}_{ts}\tilde{g})  \le \beta'_t(\chi_{ts}^E)+O(s^{-n}).
\end{equation}
Plugging the estimates \eqref{3.41}, \eqref{3.41-0} and \eqref{hatchi-ts-est} to \eqref{3.40} yields
\begin{align}\label{3.80}
\Tr(\tilde{g}\tilde{\chi}_{ts}\tilde{g}\beta_t(\rho_0)) &= O(s^{-n}).
\end{align}

Finally, recalling the definition \eqref{dX}, we find,  for all $v\in(\kappa,c)$, 
\begin{equation}\label{3.81}
\chi^\sharp_{X_{ct}^\cp}=\theta^+(d_{X_{ct}})=\theta^+(d_X-ct)\le\tilde{\chi}_{ts},
\end{equation}
where $\theta^+$ is the Heaviside function,	provided $\delta=c-v$ and $s=t$. See \figref{fig:f}.

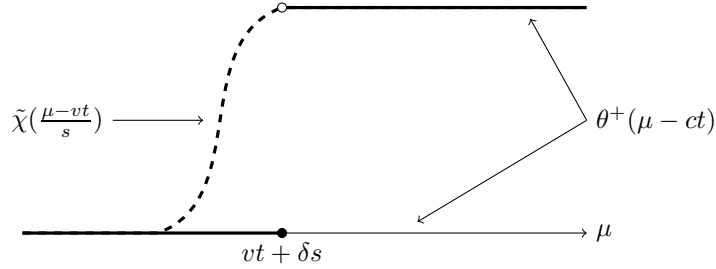
\begin{figure}[H]
\centering
\begin{tikzpicture}[scale=3]
	\draw [->] (-.5,0)--(2,0);
	\node [right] at (2,0) {$\mu$};
	
	\node [below] at (.65,0) {$vt+\delta s$};
	\draw [fill] (.65,0) circle [radius=0.02];
	

	\draw [very thick] (-.5,0)--(.65,0);
	\draw [very thick] (.65,1)--(2,1);
	\draw [dashed, very thick] (-.5,0)--(.1,0) [out=20, in=-160] to (.65,1)--(2,1);
	\filldraw [fill=white] (.65,1) circle [radius=0.02];

	\draw [->] (-.1,.5)--(.3,.5);
	\node [left] at (-.1,.5) {$\tilde \chi(\tfrac{\mu-vt}{s})$};
	
	\draw [->] (2,.5)--(1.75,.95);
	\draw [->] (2,.5)--(1.25,.05);
	\node [right] at (2,.5) {$\theta^+(\mu-ct)$};
	
\end{tikzpicture}
\caption{Schematic diagram illustrating estimate \eqref{3.81}.}
\label{fig:f}
\end{figure}

Hence we conclude estimate \eqref{1.2} from \eqref{3.80}--\eqref{3.81}. This completes the proof of \thmref{thm1}. 
\end{proof}

\medskip

\section{Estimates of multiple commutators} 
\label{sec:mult-comm-est}
In this section, we establish some key estimates for multiple commutators of the form $\ad{k}{A^E}{B}$. More precisely, we show that the operators $H=-\Delta+V(x)$ and $W_j$ satisfying \ref{H1}--\ref{W2}, with $\Phi$ given by \eqref{5.2.2}, verify that the abstract conditions \eqref{rme-cond} used to prove the recursive monotonicity estimate in Section \ref{sec:RME}.

First, we introduce some notation. For an integer $k$ and  a function $f\in C^{n+1}(\Rb^d)$, we write 
\begin{equation}\label{ak-def}
f\in \cS^k
\end{equation}
if there exists $C=C(n, f)>0$ such that for all multi-indices $\al$ with $0\le \abs{\al}\le n+1$  and $x\in\Rb^d$, 
\begin{equation}
\label{ak}
\abs{\p^\al f(x)}\le C\br{x}^{-k-|\al|}.
\end{equation}
For any multi-index $\beta$ with order $0\le|\beta|\le n+1$,  $f\in \cS^k$ and $g\in \cS^l$, it follows immediately from the definition and Leibnitz's rule that 
\begin{align}
\label{ak1}
\di^\beta f\in \cS^{k+|\beta|},\quad fg\in \cS^{k+l},
\end{align}
(with the obvious observation that $\di^\beta f\in C^{n+1-|\beta|}$ if $f\in C^{n+1}$). To simplify notation,	for a fixed operator $A$ on $\cH$, define $$C_A: B\mapsto \Ad_A (B)\equiv [A,B]$$ on the set of linear operators on $\cH$.  We also omit the subindices in $x_j$ and $p_j$. Restoring these subindices is straightforward.

\medskip
{Results in this section are valid for functions $\phi\in C^\infty, \phi\ge0$ satisfying}
\begin{equation}
\label{phi-cond}
\langle x\rangle^{\abs{\al}-1}\abs{\di^\al\phi(x)}\le M \quad  (x\in\Rb^d,\, {0\le \abs{\al}\le n+1}),
\end{equation}
for some absolute constant $M>0$.

{In particular, the smoothed-out distance function $d_X$ verifies \eqref{phi-cond}.} Later on, we choose $\phi(x)$ to be a smoothed-out distance function from $x$ to $X$, see \eqref{dX}.

The main result of this section are the following two propositions: 

\begin{proposition}\label{prop:mult-comm-estH}
Let $n\ge1$. Suppose $H$ satisfies \ref{H1} and let $\phi$ be as above.
Let $\phi^E:= g \phi g$, where $g$ is defined in \eqref{g-cond}-\eqref{g}.	Then	 there exists $C=C(n,M,E)>0$ such that, for all 
$E\in\Rb$,
\begin{align}
\label{3.15'}
\norm{\ad{k}{\phi^E}{H}} \le C\quad (k=1,\ldots,n+1).
\end{align}

\end{proposition}

\begin{proof}

1. In the following, we denote the resolvent $(z-A)^{-1}$ of the operator $A$ by $R_{A}(z)$ and $R_{A}$ if the argument is not important.  For measures, if it is clear from the context, we will also drop the arguments for simplicity.

\medskip

2. The proof is based on the mapping property of certain derivations.  Before we proceed, we define a class of operators
\begin{equation}\label{2.10}
\cF^{(1)}:=\Set{\text{ polynomials of operators of the form }B^{(1)}},
\end{equation}
where
\begin{multline}
\label{multi-comm-typ}
B^{(1)}=\int d\mu(z_1,\ldots,z_\nu)\del{\prod_{j=1}^\nu  R_{H}(z_{j})^{m_j^1}}  \del{\prod_{q=1}^{N}\prod_{r=1}^\nu a_{k_q}p^{\ell_q}R_{H}(z_{r})^{m_{r}^{q}}}, \\ 
\sum_{j=1}^\nu m_j^1\ge1,\, 0\le \ell_q \le \min(1,\sum_{r=1}^\nu m_r^{q}),\ k_q\ge 0,\ \forall q=2,\ldots N,
\end{multline}
where $\mu$ is some finite measure on $\Cb^\nu,\,\nu\ge2$, $N$ is some finite integer, and $a_k$ stands for a  {generic} function belonging to $\cS^k$ (see \eqref{ak-def}). Since $\ell_q \le \sum_{r=1}^\nu m_r^{q}$ and $k_{q}\geq 0$ for each $q$, the second factor in the integrand of \eqref{multi-comm-typ} is bounded, and therefore 
$$\cF^{(1)}\subset \cB(\cH).$$
Our goal is to show $\ad{k}{\phi^E }{H}$ lies in $\cF^{(1)}$ for all $1\le k\le n+1$ by induction, whence \eqref{3.15'} follows.

\medskip

3.  For the base case $k=1$, since $[g,H]=0$, we find by Leibnitz's rule that 
\begin{align}\label{3.7}
\ad{1}{\phi^E }{H}&=g\ad{1}{\phi  }{H} g.
\end{align}
Using formula \eqref{HSj-rep} for each $g$, we can rewrite \eqref{3.7}  using Fubini's theorem as 
\begin{equation}\label{2.11}
\ad{1}{\phi^E }{H}= \iint d\widetilde{g}(z_{1})\otimes d\widetilde{g}(z_{2})R_{H}(z_{1}) \ad{1}{\phi}{H} R_{H}(z_{2}). 
\end{equation}
By \remref{remG}, we can modify $g$ to have compact support.
Thus, we can choose the measure $d\widetilde{g}\otimes d\widetilde g$ to have compact support in $\Cb^{2}$ (see \eqref{tfDef} and Appds.~\ref{secRemEst}--\ref{4.sec:commut} for details). 

Next, we compute 					
\begin{equation}\label{3.5''}
\ad{1}{\phi}{H}=\Lap \phi+2\grad \phi\cdot \grad,
\end{equation} 
so that $\ad{1}{\phi}{H}$ is a linear combination of terms of the forms $a_1$ or $a_0p$ with $a_j\in\cS^j$, by assumption \eqref{phi-cond}.
Plugging this into \eqref{2.11} shows that $\ad{1}{\phi^E }{H}\in \cF^{(1)}$, which completes the proof of the base case. 

4.	Now, assuming 	$\ad{k}{\phi^E}{H}\in\cF^{(1)}$, we will prove $\ad{k+1}{\phi^E}{H}\in\cF^{(1)}.$  It is immediately clear that the induction step is equivalent to showing 
\begin{equation}\label{2.12}
C_{\phi^E}(\cF^{(1)})\subset\cF^{(1)}.
\end{equation}
To establish \eqref{2.12}, we use the crucial fact that the map $C_A$ is a derivation, i.e. a linear operator satisfying the Leibnitz rule. In particular, with $A= \phi^E=g\phi g=\phi g^{2}+[g,\phi]g$, we have 
\begin{equation}\label{com-repr}
C_{\phi^E}=\phi C_{g^2}+C_\phi (\cdot)g^2 + C_{\sbr{g,\phi}g}.
\end{equation}
Also, we note some easy commutator relations 
\begin{align}\label{comms-prep}
&C_{A}R_{H}=R_{H} (C_{A}H)R_{H}\quad {\text{ for all operators }A \text{ s.t. } R_H:\cD(A)\to\cD(A)},\\
&C_{H}p=i\n V,\ C_{\phi} p=i\n \phi,\	C_{\phi}H=- C_{H}\phi=\Del \phi+2\n \phi\cdot \n. \label{comms-prep'}
\end{align}

We will show that each of the three maps in \eqref{com-repr} maps $\cF^{(1)}$ into itself using the relations \eqref{comms-prep}--\eqref{comms-prep'}.

\medskip

4.1 First, we show $\phi C_{g^{2}}(\cF^{(1)})\subset\cF^{(1)}$.  Since $\phi C_{g^2}(R_{H})=0$, it suffices,  by the induction hypothesis, formula \eqref{multi-comm-typ} and Leibnitz's rule, to evaluate the operators 
\begin{equation}\label{2.14}
\phi C_{g^2}(p),\quad \phi C_{g^2}(a_{k}).
\end{equation}

Using \eqref{comms-prep}--\eqref{comms-prep'}, together with the relation
$C_{g^2}A=-\int d\widetilde{g^2}R_{H}(C_{H}A)R_{H}$ and the fact that $\n V\in \cS^{1}$ by Hypothesis \ref{H1}, we compute, using \eqref{comms-prep'}
\begin{align}
\label{comms2}\phi C_{g^2} (p)=&\int d\widetilde{g^2} \phi R_{H} (i\grad V)R_{H}\notag\\ 
=&\int d\widetilde{g^2}R_{H}a_0 R_{H}+\int d\widetilde{g^2}R_{H} (a_{1}+a_{0} p) R_{H}a_1 R_{H},
\end{align}
where in the second equality we commuted $\phi$ through $R_H$ and used again \eqref{comms-prep'} together with \eqref{phi-cond}. Similarly,
\begin{align}
\phi C_{g^2} (a_{k})=&\int d\widetilde{g^2}\phi R_{H}(\Del a_{k}+2\n a_{k}\cdot \grad ) R_{H} \notag\\
\label{comms3}=&\int d\widetilde{g^2} R_{H}(a_{k+1}+ a_{k} p) R_{H} \nonumber\\
&\quad\quad+\int d\widetilde{g^{2}}R_{H}(a_{1}+a_{0} p) R_{H}(a_{k+2}+ a_{k+1} p) R_{H},
\end{align} 
which are indeed of the desired form in order to deduce that $\phi C_{g^{2}}(\cF^{(1)})\subset\cF^{(1)}$.

\medskip

4.2 Next, we show $C_{\phi}(\cF^{(1)})g^{2}\subset\cF^{(1)}$.  Since $C_{\phi}(a_{k})=0$ for all $k$, it suffices,  by induction hypothesis, formula \eqref{multi-comm-typ} and Leibnitz's rule, to evaluate the following operators 
\begin{equation}\label{2.13}
C_{\phi}(p),\quad C_{\phi}(R_{H}),
\end{equation}
where, recall, $R_{H}$ stands for the resolvent of $H$.  Using  the relations \eqref{comms-prep}--\eqref{comms-prep'}, we compute
\begin{align}
C_\phi (p)=&\grad \phi\in \mathcal{S}^0,			\label{2.16}\\
C_{\phi}(R_{H})&= R_{H}(C_{\phi}H)R_{H}=R_{H}(\Del \phi+2\n \phi\cdot \grad)R_{H}\nonumber\\
&=R_{H}(a_{1}+a_{0} p)R_{H},\label{2.17}
\end{align}
which, inserted into \eqref{multi-comm-typ}, allows us to conclude that $C_{\phi}(\cF^{(1)})g^{2}\subset\cF^{(1)}$.

\medskip

4.3	 Finally, we  show $C_{[g,\phi]g}(\cF^{(1)})\subset\cF^{(1)}$. By the induction hypothesis and the Leibnitz rule, it suffices to show that $\sbr{g,\phi}g$ is of the form \eqref{multi-comm-typ}.	To this end we use \eqref{HSj-rep} so that
\begin{equation}\label{3.5'}
\begin{aligned}
	\sbr{g,\phi}g&=\del{\int d \widetilde{g}(z_{1})\sbr{R_{H}(z),\phi}}\del{\int d \widetilde{g}(z_{2})R_{H}(z_{2})}\\
	&= -\del{\int d \widetilde{g}(z_{1})R_{H}(z_{1})\del{\ad{1}{\phi}{H}}R_{H}(z_{1})}\del{\int d \widetilde{g}(z_{2})R_{H}(z_{2})}.
\end{aligned}
\end{equation}
Since $C_{\phi}(H)=a_1+a_0p$ from \eqref{comms-prep'}, Eq.~\eqref{3.5'} shows that  $C_{\sbr{g,\phi}g}(\cF^{(1)})\subset \cF^{(1)}$. 

This completes the induction.
\end{proof}

\begin{proposition}\label{prop:mult-comm-estW} 
Suppose Assumption \ref{W2} holds and let  $\phi\in C^\infty(\Rb^d)$ satisfy condition  \eqref{phi-cond}. 
Let $\phi^E=g \phi g$ where $g$ is defined in \eqref{g-cond}-\eqref{g}. 
Then, the following estimates hold: 
\begin{equation}\label{commWj-est}
\sum_j\norm{\ad{k}{\phi^E}{W_j}}^{2} <\infty\quad (k=0,\ldots, n+1).
\end{equation}
\end{proposition}
\begin{proof} 

Within this proof we fix some $j$ and write $W\equiv W_j$.   We will use the same strategy and adapt the same notations in the proof of \propref{prop:mult-comm-estH} to establish mapping property for the derivation $C_{\phi^E}$.  For each $k=1,\ldots, n+1$, we define the classes of operators on $\cB(\cH)$
\begin{align}\label{2.19}
\cG^{(2)}_{m}&:=\{\cL_{A}\cR_{A'}B_{rs}^{(2)}\mid A,A'\in\cF^{(1)}\cup\{\textbf{1}\},\nonumber\\
&\quad\quad\quad\quad\quad\quad\quad B^{(2)}_{rs}\equiv (\phi C_{p})^{r}C_{x}^{s} \text{ with }r,s\geq 0\text{ and }r+s=m\}\nonumber\\
\cF^{(2)}_k&:=\Big\{\text{ polynomials of elements in }\cG^{(2)}_{m}(W)\text{ for }1\leq m\leq k\Big\}.
\end{align}
Here $\cL,\cR$ are left- and right-multiplication operator in $\cB(\cH)$, respectively, $\cF^{(1)}$ is defined in \eqref{2.10}, and $\cG_{m}^{(2)}(W)$ means operators in $\cG_{m}^{(2)}$ acting on $W$.

\medskip

1. Our first claim is that
\begin{equation}\label{2.20}
\ad{k}{\phi^E}{W}\in \cF^{(2)}_k
\end{equation}
for every $k=1,\ldots,n+1$.  We prove the this claim by induction in $k$. For $k=1$, we first compute
\begin{align}\label{com-H-W}
C_{H}W&=pC_{p}W+(C_{p}W)p+C_{V}W\nonumber\\
&=pC_{p}W+(C_{p}W)p+\int d\widetilde{V}(z)R_{x}(z)[C_{x}W]R_{x}(z)
\end{align}
and
\begin{align}\label{com-phi-V-W}
\phi C_{V}W&=\int d\widetilde{V}(z)\phi(x)R_{x}(z)[C_{x}W]R_{x}(z)\nonumber\\
&=a_{0}\int d\widetilde{V}(z)(\textbf{1}-(z-i)R_{x}(z))[C_{x}W]R_{x}(z),
\end{align}
using the identity $R_{x}(z)=(x-i)^{-1}[\textbf{1}-(z-i)R_{x}(z)]$ and noting that $\phi (x-i)^{-1}\in \cS^{0}$. Note that the integral in \eqref{com-phi-V-W} is convergent, as follows from the fact that $V\in\mathcal{S}^\rho$ for some $\rho>0$ (see Hypothesis \ref{H1}) together with the properties of the almost analytic extension $\widetilde V$ described in Appendix \ref{secRemEst}. 

Here and below, to simplify the proof we take $d=1$. For $d\ge1$, we use the  Helffer-Sj\"ostrand   representation \eqref{HSj-rep} for several variables to write 
$$V(x_1,\ldots,x_d)=\int  d\widetilde{V}(z_1,\ldots, z_d)(z-x_1)^{-1}\ldots(z-x_d)^{-1},$$
which yields through Leibnitz rule that
\begin{align*}
\phi C_{V}W&=\int d\widetilde{V}(z)\phi(x)R_{x_1}(z)[C_{x_1}W]R_{x_1}(z)R_{x_2}(z)\cdots R_{x_d}(z)+\ldots\\
&\qquad + \int d\widetilde{V}(z)\phi(x)R_{x_1}(z)\cdots R_{x_d}(z)[C_{x_d}W]R_{x_d}(z). 
\end{align*}
One can handle each of the $d$ terms on the r.h.s. exactly as in \eqref{com-phi-V-W} and then sum over the results.

Eqs. \eqref{com-H-W}--\eqref{com-phi-V-W} show that $\phi C_{H}W\in\cF_{1}^{(2)}$.  Now, using \eqref{com-repr}--\eqref{comms-prep'} and that fact that $g^{2},[g,\phi]g\in \cF^{(1)}$, as shown in Proposition \ref{prop:mult-comm-estH}, we have
\begin{align}
\phi C_{g^{2}}(W)&=\int d\widetilde{g}(z)R_{H}(z)\phi C_{H}(W)R_{H}(z)\nonumber\\
&\quad\quad+\int d\widetilde{g}(z)R_{H}(z)(a_{1}+a_{0}p)R_{H}(z)[C_{H}W]R_{H}(z)\\
C_{\phi}(W)g^{2}&=\int d\widetilde{\phi}(z)R_{x}(z)[C_{x}W]R_{x}(z)g^{2}(H),\label{2.40}\\
C_{[g,\phi]g}W&=[g,\phi]g W-W[g,\phi]g,
\end{align}
so that $C_{\phi^{E}}W\in\cF_{1}^{(1)}$.  This completes the proof for the base case.

\medskip

2. Now,  assuming \eqref{2.20} holds for $k=m$, we prove it for $k=m+1$.  Since $\ad{m+1}{\phi^E}{W_j}=C_{\phi^E}(\ad{m}{\phi^E}{W_j})$, by inductive assumption, it suffices to show that $C_{\phi^{E}}(AB_{m}A')\in\cF_{m}^{(2)}$ for all $AB_{m}A'\in \cG_{m}^{(2)}$.  By Leibnitz rule, we have
\begin{align}\label{ind-Westi}
C_{\phi^{E}}(AB_{m}A')&=(C_{\phi^{E}}A)B_{m}A'+A(C_{\phi^{E}}B_{m})A'+AB_{m}(C_{\phi^{E}}A').
\end{align}
The first and the last term on the r.h.s. of \eqref{ind-Westi} is taken cared by Proposition \ref{prop:mult-comm-estH}.  We now have to compute the second term.  To this end, we define another set of operators
\begin{align}\label{2.26}
\cG^{(3)}_{m}&:=\{\cL_{A}\cR_{A'}B_{rs}^{(3)}\mid A,A'\in\cF^{(1)}\cup\{\textbf{1}\},B^{(3)}_{rs}\equiv (\phi^{\ell} C_{p})^{r}C_{x}^{s}\nonumber\\
&\quad\quad\quad\quad\quad\quad\quad \text{ with }\ell\in \{0,1\}, r,s\geq 0\text{ and }r+s=m\}\nonumber\\
\cF^{(2)}_k&:=\Big\{\text{polynomials of elements in }\cG^{(3)}_{m}(W)\text{ for }1\leq m\leq k\Big\}.
\end{align}
We remark that the operator product $(\phi^{\ell}C_{p})^{r}$ means that products of the form $(\phi C_{p})^{r_{1}}(C_{p})^{r_{2}}...(\phi C_{p})^{r_{2n-1}}(C_{p})^{r_{2n}}$ for any $r_{1},...,r_{2n}\geq 0$ and $r_{1}+...+r_{2n}=r$.

Write $\phi=b\br x$ with $b(x):=\phi(x)/\br{x}\in \cS^0$ by \eqref{phi-cond} with $\al=0$.  We successively commute the bounded operators $b$'s to the left. Then condition \eqref{W2-cond'} implies the same estimate but with $\phi$ in place of $\br{x}$, i.e.
\begin{align}\label{W2-cond}
\sum_{j=1}^\infty\sum_{\substack{\sum(k_i+\ell_i)=n+1\\ k_i,\,\ell_i\ge0}}\|\prod_{i} \big[(\phi C_{p_q})^{k_i}C_{x_q}^{\ell_i} W_j\big]\|^2<\infty.
\end{align}
By \eqref{W2-cond} and the fact that $\cF^{(1)}\subset \cB(\cH)$, it follows that 
$$\cF^{(3)}_{n+1}\subset \cB(\cH).$$
We now claim that for $k=0,1\ldots$ and every $B^{(2)}_k\in\cG^{(2)}_k$, there exist $$A,\,A'\in\cF^{(1)}\cup\Set{\1},\quad B_k^{(3)}(W)\in\cF^{(3)}$$ 
such that 
\begin{equation}\label{2.28}
B_k^{(2)}(W)=A B_k^{(3)}(W) A'.
\end{equation}
This relation implies $\cF^{(2)}\subset \cF^{(3)}$. 
This, together with \eqref{2.20} and the inclusion $\cF^{(3)}_{n+1}\subset \cB(\cH)$, leads to \eqref{commWj-est}.

\medskip

2.1 Again, we prove \eqref{2.28} by induction.  For $k=1$, it is trivial from the definition.

\medskip

2.2 Next,  assuming \eqref{2.28} holds for $k=m$, we prove it for $k=m+1$.  Again, by Proposition \ref{prop:mult-comm-estH} and by the induction hypothesis and the Leibnitz rule, it suffices therefore to show that, for any $B^{(3)}_{m}=(\phi^{\ell}C_{p})^{r}C_{x}^{s}$ for some $\ell\in\{0,1\}$ and $r,s\geq 0$ such that $r+s=m$, 
\begin{equation}\label{2.31}
\phi C_{p}(B_{m}^{(2)}(W)),C_{x}(B_{m}^{(2)}(W))\in\cF^{(3)}_{m+1}.
\end{equation}
For the former term, it is trivial.  For the latter case, we use the fact that
\begin{align}
C_{p}C_{x}=C_{x}C_{p},\quad\quad C_{x}(\phi C_{p})=\phi C_{p}C_{x}
\end{align}
so that $C_{x}(B_{m}^{(3)}W)=B_{m}^{(3)}C_{x}W=(\phi^{\ell}C_{p})^{r}C_{x}^{s+1}W$.  This completes the induction.

\medskip

3.  Now we return back to our previous induction proof.  Since every operator in $\cG_{m}^{(2)}(W)$ can be expanded as a finite sum of terms in $\cF^{(3)}_{m}$, it suffices to calculate $C_{\phi^{E}}(B^{(3)}_{m}(W))$ for some $B^{(3)}_{m}=(\phi^{\ell}C_{p})^{r}C_{x}^{s}\in\cG^{(3)}_{m}$.  As in the calculation for the base case, it suffices to compute the terms $\phi C_{H}(B_{m}^{(3)}(W))$ and $C_{x}(B_{m}^{(3)}(W))$.  The latter term is contained in $\cF^{(3)}_{m+1}$ trivially.  For the former term, we have
\begin{align}\label{comm-H-B3}
\phi C_{H}(B_{m}^{(3)}(W))&=\phi pC_{p}(B_{m}^{(3)}(W))+\phi C_{p}(B_{m}^{(3)}(W))p+\phi C_{V}(B^{(3)}_{m}(W))\nonumber\\
&=p(\phi C_{p}B^{(3)}_{m}(W))+(\phi C_{p}B^{(3)}_{m}(W))p\nonumber\\
&\quad\quad\quad+(C_{p}\phi) B^{(3)}_{m}(W)+\phi C_{V}(B^{(3)}_m(W)).
\end{align}
Obviously the first three terms in the last line of \eqref{comm-H-B3} belong in $\cF^{(3)}_{m+1}$.  For the last term, we have 
\begin{align}\label{comm3}
\phi C_{V} B^{(3)}_m&=a_{0}\int d\widetilde{\phi}(z)[\one - (z-i)R_{x}(z)](C_{x} B^{(3)}_{m}(W))R_{x}(z) \notag\\
&= a_{0}(C_{x} B^{(3)}_{m}(W))a_0-\int d\widetilde{\phi}(z)(z-i)   R_{x}(z)(C_{x} B^{(3)}_{m}(W))R_{x}(z) .
\end{align}  
This completes the proof.
\end{proof}

\begin{corollary}\label{cor:4.3}
Suppose that $H=-\Delta+V(x)$ and $W_j$ satisfy \ref{H1}--\ref{W2}. Then, with $\Phi$ given by \eqref{5.2.2},  condition \eqref{rme-cond} holds.
\end{corollary}	
\begin{proof}
Since $d_X(x)$ satisfies condition \eqref{phi-cond}, it suffices to apply Propositions \ref{prop:mult-comm-estH}--\ref{prop:mult-comm-estW}.
\end{proof}

\medskip

		\section{Proof of  claims \eqref{claim1}--\eqref{claim2} }\label{sec:claims-deriv}

\subsection{Proof of Claim  \eqref{claim1}}\label{secPfCl1}
Recall that $\chi_X^\sharp,\,X\subset \Rb^d$, denotes the characteristic functions of $X$. Recall also that the set of smooth cutoff functions $\mathcal{X}$ is defined in \eqref{F} and that $d_X^E=gd_Xg$ with $g=g^E(H)$ (see \eqref{g-cond}--\eqref{gE}) and $d_X$ the smooth distance function defined in \eqref{dX}. We reproduce Claim \eqref{claim1} below:
\begin{proposition}\label{cl:1}	For  every $\chi\in\cX$ and $\chi_{0s}=\chi(s^{-1}d_X^E)$ (see \eqref{chitsdef}), 		
	\begin{align}
		\label{claim1'} \chi_X^\sharp\chi_{0s}\chi_X^\sharp&=O(s^{-n}).
	\end{align}
\end{proposition}
\begin{remark}
	This is a semiclassical estimate which physically says that a quantum particle that is essentially localized in phase space inside an energy ball and outside of $X$ (by way of $d_X^E$) is also localized outside of $X$ in position space up to small errors. A technical challenge here is that the operator $d_X$ is unbounded.
\end{remark}

\begin{proof}[Proof of \propref{cl:1}]
	In the remainder of this proof, we use the following notations: For $z\in \Cb$ with $\Im(z)\ne0$,  $d$ as in \eqref{dX}, and $g$ as in \eqref{g},   \[\begin{aligned}
		d\equiv& d_X,\qquad d^E\equiv d^E_X=gd_Xg,\qquad R=(d/s-z)^{-1},\qquad R^E=(d^E/s-z)^{-1},\\ b=&d-d^E,\qquad \chi^E=\chi(d^E/s),\qquad 
		\chi=\chi(d/s).
	\end{aligned}
	\]

	We begin with
	\begin{lemma}\label{lem:Rb-bdd}
		The operator $Rb$ is bounded.
	\end{lemma}
	\begin{proof}
		Since $b=d-d^{E}$ and $Rd$ is bounded as the multiplication operator by a bounded function, it suffices to show that $Rd^{E}$ is bounded.  For the latter, we have, by \eqref{gE}, 
		\begin{align}
			Rd^{E}&=Rgdg=Rdg^{2}+R [g,d]g.
		\end{align}
		Since $g$ is bounded and $Rd=s(1+zR)$ so that
		$\|Rd\|\le s(1+\abs{z}\abs{\Im(z)}^{-1})$,
		it remains to show that $[g,d]$ is bounded.  Using the HS representation \eqref{HSj-rep} with $k=0$ and formula \eqref{3.5''}, we have	
		\begin{align}
			[g,d]=& \int d \widetilde{g}(z)\sbr{(z-H)^{-1},d}\notag
			\\=&-{\int d \widetilde{g}(z)(z-H)^{-1}{\ad{1}{d}{H}}(z-H)^{-1}}\notag
			\\
			=&\int d \widetilde{g}(z) (z-H)^{-1}\del{\grad \cdot(\grad d)+ \grad d\cdot \grad}(z-H)^{-1}
			\label{gdcomm}.
		\end{align}
		
		Next we multiply by $i$ and use the operator Cauchy-Schwarz inequality
		\[
		\begin{aligned}
			&i\grad \cdot(\grad d)+\grad d\cdot i\grad
			\\\leq& -\br{E}^{-1/2}\Delta+\br{E}^{1/2}\abs{\nabla d}^2
			\\\leq& \frac{H}{\br{E}^{1/2}}+\|V\|_\infty+1+\br{E}^{1/2}\abs{\nabla d}^2 =: B_{H,E}.
		\end{aligned}
		\]
		By \eqref{dX2}, we have $|\nabla d|\leq C$. This, together with condition \eqref{V-cond} on $V$ and the  HS representation \eqref{HSj-rep} with $k=1$, shows that
	
			\begin{align}
				&\big\| (z-H)^{-1}\del{\grad \cdot(\grad d)+ \grad d\cdot \grad}(z-H)^{-1} \big\| \\
				&\le \big\| B_{H,E}^{\frac12} (\bar z - H )^{-1} \big\|\big\| B_{H,E}^{\frac12} ( z - H )^{-1}\big\|\\
				&\le C\big ( \langle E \rangle^{-\frac12} |z|+ \langle E \rangle^{\frac12} |\mathrm{Im}(z)|^{-2} \big ).
			\end{align}					
			Using the properties of the almost analytic extension $\tilde g$ (in particular the fact that it is compactly supported, see \eqref{tfDef} and \remref{remG}), this shows that the integral in \eqref{gdcomm} is norm convergent, which completes the proof.
						\end{proof}
		
	Now, using the Helffer-Sj\"ostrand representation  \eqref{HSj-rep} and omitting the measure, we write
	\begin{align}\label{3.11}
		\chi^E=\int R^E.
	\end{align}
	Using that the operator $Rb$ is bounded and expanding $R^E=(d^E/s-z)^{-1}=(d/s-z-b/s)^{-1}$ in powers of $R b/s$ up to the order $n-1$, we obtain
	\begin{align}
		R^E&=(d/s-z-b/s)^{-1}=\sum_{k=0}^{n-1}s^{-k}(Rb)^kR+s^{-n}(Rb)^nR^E.
	\end{align}
	Plugging this expansion into \eqref{3.11} yields
	\begin{align}\label{Ak-part-exp}
		\chi^E&=\sum_{k=0}^{n-1} \chi_k+s^{-n}\Rem_1,
	\end{align}
	where
	\begin{align}\label{Ak-part-exp2}
		\chi_k&=\int (Rb/s)^kR\quad\text{and}\quad \Rem_1=\int (Rb)^{n}R^E.
	\end{align}

	Our goal is to move the $R$'s in the first integrand to the right. Using the relations $Rb=bR+[R, b]$ and $[R,b]=-s^{-1}R \ad{}{d}{b}R$, we would like to   obtain an expansion of the form
	\begin{align}\label{Rbk-exp-wrong}
		(Rb)^k R &=\sum_{ l}s^{-i_l} \tilde B_{l}R^{l +1}+s^{-n}\tilde M_{k},
	\end{align}
	where the operators   $ \tilde B_{l}$ and $\tilde M_{k}$ are  polynomials of operators $\ad{k}{d}{b}, k=0, 1, \dots ,$  (and $R$ for $\tilde M_{k}$), and then use  $\int R^{l+1}=(-1)^{l+1}\chi^{(l)}$ (see \eqref{HSj-rep}) and  $\chi^{(l)}\chi_X^\#=0$ for all $ l\geq 0$. 			The problem here is that the operators $\ad{k}{d}{b} $  are not bounded, so  $ \tilde B_{l}$ and $\tilde M_{k}$ are not guaranteed to be bounded operators.	
	Hence, we proceed differently.	
	
	We transform the product $(Rb/s)^k$ as follows. We  use the relation 
	\begin{align}\label{b-d}
		b=&gd\bar{g}+\bar{g}d=d h-\ad{}{d}{\bar g}g,\\
		\label{h}	\text{where }&\ \bar{g}=1-g \text{ and } h:=\bar{g}(1+g),
	\end{align}
	and the definition $R=(d/s-z)^{-1}$  to write 
	\begin{align}\label{Rb-transf}
		Rb/s &=d_sR h+Rc_s,\text{ where } \\
		d_s&:=d/s,\quad \label{Rb-transf2}c:=\Ad_{d}(g)g,\quad c_s=c/s. 
	\end{align}

	Notice that the operators $c_s,\,h$ and $d_sR$ are bounded and  
	\begin{align}\label{R1-expr}
		d_sR=\one +z R.
	\end{align}
	The last two relations imply 
	\begin{align}\label{Rb-transf3}
		Rb/s&=h+Rc_s+zRh.
	\end{align}
	Our goal is to move the $R$'s to the extreme right to obtain the following:
	
	\begin{lemma}\label{lem:(Rb/s)k-exp}
		The operator $(Rb/s)^{k}$ has the following expansion:
		\begin{align}\label{Rbk-exp}
			(R b/s)^{k} =h^k+\sum_{q=0}^k\sum_{l=0}^{n-1}s^{-l} B_{q,l}R^{l+1}p_{q,l}(z)+s^{-n}\sum_{q=0}^kM_{q,n} p_{q,n}(z),
		\end{align}
		where
		\begin{enumerate}
			\item[(a)] 
				$k=1,...,n-1$,

				\item[(b)] the operators $ B_{q,l}$ are  polynomials of bounded operators  $\ad{m}{d}{h}$ and $\Ad_{d}^{m}(c_s)$, with $0\leq m\leq l$, 
				\item[(c)] the operators $M_{q,n}$  are polynomials of bounded operators $R$, $\ad{m}{d}{h}$ and $\Ad_{d}^{m}(c_s)$, with $0\le m\leq n$ and
				\begin{align}\label{degMbdd}
					\deg_R(M_{q,n}):=\text{powers of $R$ in $M_{q,n}$}\in[n+1,n+k].
				\end{align}				
				
				\item[(d)] $p_{q,l}(z)$ are polynomials in $z$ of the degree $\le q$. 
			\end{enumerate}
		\end{lemma}
		We call the operators described in (b) as \textit{$l$-operators}.  Note that if $B_l$ is an $l$-operator, then it is also an $(l+m)$-operator for $m\geq 1$.
		\begin{remark}
			The negative powers of $s$ come from the commutator relation
			\begin{align}\label{R1h-comm} 
				[R, B]=-s^{-1}R \Ad_d(B)R,
			\end{align}		
			valid for any bounded operator $B$ and $\Im(z)\ne0$.
		\end{remark}

		\begin{proof}[Proof of Lemma \ref{lem:(Rb/s)k-exp}]
			We prove  \eqref{Rbk-exp} by induction on $k$.  
			
			For the base case $k=1$, we use the commutator expansion
			\begin{align}\label{comm-exp}
				RB=\sum_{r=0}^{p-1}(-1)^rs^{-r}\Ad_{d}^r(B) R^{r+1}+(-1)^ps^{-p}R\Ad_{d}^p(B) R^{p},
			\end{align}
			valid for any bounded operators $B$ and integer $p\ge1$. Applying \eqref{comm-exp} to $B=h$ and $c_s$ (see \eqref{Rb-transf}), we find
			\begin{align}
				Rb/s=&h+Rc_s+zRh \notag
				\\=&h+\sum_{r=0}^{n-1}(-1)^rs^{-r}\Ad_{d}^r(c_s) R^{r+1}+(-1)^ns^{-n}R\Ad_{d}^n(c_s) R^{n}\notag
				\\&+z\del{\sum_{r=0}^{n-1}(-1)^rs^{-r}\Ad_{d}^r(h) R^{r+1}+(-1)^ns^{-n}R\Ad_{d}^n(h) R^{n}}.
			\end{align}
			This is of the form \eqref{Rbk-exp} with
			\begin{align}
				&B_{0,r}:=(-1)^r\Ad_{d}^r(c_s), \quad M_{0,n}:=(-1)^nR\Ad_{d}^n(c_s) R^{n},\\
				&B_{1,r}:=(-1)^r\Ad_{d}^r(h), \quad M_{1,n}:=(-1)^nR\Ad_{d}^n(h) R^{n},
			\end{align}
			where
			\begin{align}
				\deg_R(M_{0,n})=\deg_R(M_{1,n})=n+1
			\end{align}
			satisfies \eqref{degMbdd}.
			
			Now we assume \eqref{Rbk-exp} for a given $k\geq 1$ and prove it for $k\rightarrow k+1$.  We use \eqref{Rb-transf3}  to write 
			\begin{align}\label{Rbk-ABC}
				(Rb/s)^{k+1}&=(zRh+Rc_s+h)^{k+1}\nonumber\\
				&=zRh(Rb/s)^{k}+Rc_s(Rb/s)^{k}+h(Rb/s)^{k}\nonumber\\
				&=:A+B+C.
			\end{align}
			Using the induction hypothesis, we see that the third term on the r.h.s. of \eqref{Rbk-ABC} is already in the desired form (notice that the term $h^{k+1}$ in \eqref{Rbk-exp}  comes from this contribution). The first two terms on the r.h.s. of \eqref{Rbk-ABC} are treated similarly, so we only consider the first term.  
			
			We transform the term $A$ in line \eqref{Rbk-ABC} as
			\begin{align}
				A=&A_1+A_2+A_3,\end{align}
			where \begin{align}A_1:=&zRh^{k+1},\label{Rbk-A}
				\\A_2:=&zRh\sum_{q=0}^k\sum_{l=0}^{n-1}s^{-l} B_{q,l}R^{l+1}p_{q,l}(z),
				\\A_3:=&	s^{-n}zRh\sum_{q=0}^kM_{q,n}p_{q,n}(z).		
			\end{align}
			The term $A_1$ can be handled using expansion \eqref{comm-exp} as
			\begin{align}
				A_1=z\del{\sum_{l=0}^{n-1}(-1)^ls^{-l}\Ad_{d}^l(h^{k+1}) R^{l+1}+(-1)^ns^{-n}R\Ad_{d}^n(h^{k+1}) R^{n}}.
			\end{align}
			By Leibniz's rule, for each $l$, $\Ad_{d}^l(h^{k+1})$ is an $l$-operator as defined in part (b) of Lemma \ref{lem:(Rb/s)k-exp}, and so $A_1$ is of the form \eqref{Rbk-exp} with
			\begin{align}
				&B^{(1)}_{1,l}:=(-1)^ls^{-l}\Ad_{d}^l(h^{k+1})   ,\quad p^{(1)}_{q,l}:=\delta_{1q}z, \\ &M^{(1)}_{1,n}:=(-1)^ns^{-n}R\Ad_{d}^n(h^{k+1}) R^{n}\text{ satisfying } \deg_R(M^{(1)}_{1,n})=n+1.
			\end{align}

			The term $A_3$ can be written as
			\begin{align}
				A_3= \sum_{q=0}^k(RhM_{q,n})(zp_{q,n}(z))=\sum_{q=1}^{k+1}M^{(2)}_{q,n }p^{(2)}_{q,n },
			\end{align}
			where 
			\begin{align}
				M^{(2)}_{q ,n }:=RhM_{q-1,n},\quad p^{(2)}_{q ,n }:=zp_{q-1,n}(z),
			\end{align}
			with notations as in parts (c)-(d) of Lemma \ref{lem:(Rb/s)k-exp}. Since $\deg_R M_{q,n}\le n+k$, we have
			\begin{align}
				\deg _R 	M^{(2)}_{q ,n } \in [n+2,n+k+1],
			\end{align}		
			which satisfies the bound \eqref{degMbdd} with $k\to k+1$.	
			Thus $A_3$ is of the form \eqref{Rbk-exp}.

			To bring the term $A_2$ into the desired form, we commute $R$'s in \eqref{Rbk-A} to the right using expansion \eqref{comm-exp}. For each $q=0,\ldots, k$, we consider the sum
			\begin{align}\label{A3q}
				A_2(q):=\sum_{l=0}^{n-1}s^{-l} zRhB_{q,l}R^{l+1}p_{q,l}(z),
			\end{align}
			so that \begin{align}\label{A3exp}
				A_2=\sum_{q=0}^k A_2(q).
			\end{align}
			Let $B'_{q,l}=hB_{q,l}$.  Using \eqref{comm-exp}, we have, for each $l=0,\ldots, n-1$,  
			\begin{align}\label{Rbk-A3}
				RhB_{q,l}R^{l}&=RB'_{q,l}R^{l}\\
				&=\sum_{r=0}^{n-l-1}(-1)^rs^{-r}\Ad_{d}^{r}(B'_{q,l})R^{l+r+1}+(-1)^{n-l}s^{-(n-l)}R\Ad_{d}^{n-l}(B'_{q,l})R^{n}.\nonumber
			\end{align}
			Using Leibniz rule for commutators and the structure of $B_{q,l}$, we conclude that the operators $\Ad_{d}^{r}(B'_{q,l})$ are polynomials of $\Ad_{d}^{m}(h)$ and $\Ad_{d}^{m}(c_s)$, $m=0,1,...,l+r$, and therefore are $(l+r)$-operators as defined above.  So, setting $B_{q,l+r}''=(-1)^r\Ad_{d}^{r}(B_{q,l}')$, expansion \eqref{Rbk-A3} becomes
			\begin{align}\label{Rbk-A2}
				RhB_{q,l}R^{l}	=
				\sum_{r=0}^{n-l-1}s^{-r}B''_{q,l+r}R^{l+r+1}+s^{-(n-l)}RB''_{q,n}R^{n}.
			\end{align}
			Substituting \eqref{Rbk-A2} into \eqref{A3q} and setting $p_{q+1,l}'(z):=zp_{q,l}(z)$ for $l=0,\ldots,n-1$, we obtain
			\begin{align}
				A_2(q) =&\sum_{l=0}^{n-1}\sum_{r=0}^{n-l-1}s^{-(l+r)}	 B''_{q,l+r}R^{l+r+1}p'_{q+1,l}(z)\label{Rbk-A4}
				\\&+s^{-n}\sum_{l=0}^{n-1}RB_{q,n}''R^{n}p'_{q+1,l}(z).\notag
			\end{align}
			Changing the summation index $(l+r,l)\to (l',r')$, the r.h.s. in line \eqref{Rbk-A4} can be written as
			\begin{align}\label{coi}
				\sum_{l=0}^{n-1}\sum_{r=0}^{n-l-1}s^{-(l+r)}	 B''_{q,l+r}R^{l+r+1}p'_{q+1,l}(z)=\sum_{l'=0}^{n-1}\sum_{r'=0}^{l'}s^{-l'}	 B''_{q,l'}R^{l'+1}p'_{q+1,r'}(z).
			\end{align}
			Setting $p''_{q+1,n}:=\sum_{l=0}^{n-1}p'_{q+1,l}(z)$ in \eqref{Rbk-A4} and $p''_{q+1,l'}:=\sum_{r'=0}^{l'}p'_{q+1,r'}$  for each $l'=0,\ldots,n-1$ in \eqref{coi}, we conclude that
			\begin{align}
				A_2(q) =&\sum_{l=0}^{n-1}s^{-l}	 B''_{q,l}R^{l+1}p''_{q+1,l}(z)\notag
				\\&+s^{-n}RB_{q,n}''R^{n}p''_{q+1,n}(z).\label{536}	
			\end{align}
			Plugging \eqref{536} into \eqref{A3exp} yields
			\begin{align}
				A_2=&\sum_{q=0}^k\sum_{l=0}^{n-1}s^{-l}	 B_{q,l}''R^{l+1}p''_{q+1,l}(z)\notag
				\\&+s^{-n}\sum_{q=0}^kRB_{q,n}''R^{n}p''_{q+1,n}(z)
			\end{align}			
			Shifting the dummy index $q\to q+1$ and setting \begin{align}
				&B^{(3)}_{q,l}:=B''_{q-1,l+1},\quad p^{(3)}_{q,n}(z):=p''_{q+1,n}(z),  \\&M^{(3)}_{q,n}:=RB''_{q-1,n}R^n\text{ with }\deg_R(M^{(3)}_{q,n})=n+1,
			\end{align}	
			we conclude that $A_2$ is of the form \eqref{Rbk-exp}.

			This completes the proof of \lemref{lem:(Rb/s)k-exp}. 
		\end{proof}
		
		\begin{corollary}\label{corChik}
			For any $\chi\in C^\infty(\Rb)$ with compactly supported derivative and $\chi_k=\int (Rb/s)^kR\,d\tilde\chi(z)$,
			\begin{align}\label{chik-chiX}
				\chi_k =&h^k\chi(d_s)+\sum_{q=0}^k\sum_{l=0}^{n-1}s^{-l} B_{q,l}(\chi p_{q,l})^{(l+1)}(d_s)+s^{-n}\Rem_{2,k},
			\end{align}	
			where $B_{q,l}$ are as in Lemma \ref{lem:(Rb/s)k-exp}  and $\Rem_{2,k}=O(1)$. 	
		\end{corollary}
		\begin{proof}
			We have by the Heffler-Sj\"orstrand representation \eqref{HSj-rep} that   $\int R^{l+1}p_{l}(z)=(-1)^{l+1}(\chi p_{l})^{(l)}(d_s)$ (see \eqref{HSj-rep}).
			
			This, together with the definition $\chi_k =\int (Rb/s)^kR$ and expansion \eqref{Rbk-exp}, implies 
			\begin{align}\label{chik-chiX'}
				\chi_k =&h^k\chi(d_s)+\sum_{q=0}^k\sum_{l=0}^{n-1}s^{-l} B_{q,l}(\chi p_{q,l})^{(l+1)}(d_s)\notag\\&+s^{-n}\sum_{q=0}^k\int M_{q,n}Rp_{q,n}(z)\,d\tilde\chi(z).
			\end{align}		
			Thus it remains to show the integral on line \eqref{chik-chiX'} is $O(1)$.
			
			Using the estimate $\norm{R}\le \abs{\Im(z)}^{-1}$  and the degree bound \eqref{degMbdd} and that $k\le n-1$, 
			we    have \begin{align}\label{MnEst}
				\norm{ M_{q,n}}\le C\sum _{j=n}^{2n}\abs{\Im(z)}^{-(j+1)}\text{for  all $q$}.
			\end{align} 
			Since $p_{q,n}$ has degree at most $n$ and $\tilde \chi$ has compactly supported derivatives, we find  by expression \eqref{MnEst} and  \corref{corRempEst} with $(p,l)=(n+1,n),\ldots, (2n+1,n)$ 					that

			\begin{align}\label{MRemest}	
				\norm{\int M_{q,n}R p_{q,n} \,d\widetilde {\chi}(z)}\le C\int  \sum_{j=n}^{2n}\abs{\Im(z)}^{-(j+2)}\,\abs{p_{q,n}(z)} \,d\widetilde {\chi}(z)\le C.
			\end{align}
			Summing \eqref{MRemest} over $q$ shows that the integral on line \eqref{chik-chiX'} is $O(1)$. This completes the proof of \corref{corChik}.
		\end{proof}

		Since  $\chi^{(l)}(d_s)\chi_X^\#=0$ for all $ l\geq 0$, expansion \eqref{chik-chiX} gives  				
		\begin{align}\label{chikEst}
			\chi_k\chi_X^\#&=s^{-n}\Rem_{2,k}\chi_X^\sharp=O(s^{-n}) .
		\end{align}		
		Next, we deal with the $\Rem_1$ term in \eqref{Ak-part-exp}. We use the splitting
		\begin{align}\label{Rb-transf'}
			Rb  &=Rc+R_2 h,\quad   c:=\Ad_{d}(g)g, \quad R_2:=dR,		
		\end{align}
		which follows from \eqref{Rb-transf}. We prove:
		\begin{lemma}\label{lemRbnExp}
			For $k\ge1$,	the operator $(Rb)^{k}$ has the following expansion:
			\begin{align}\label{RbnExp}
				(R b)^{k} =&\sum_{l=0}^k R_2^l N_{k-l},
			\end{align}
			where the operators $N_j$  are polynomials of bounded operators $R$, $\ad{m}{d}{h}$ and $\Ad_{d}^{m}(c)$  with $0\le m\leq k-1$ and 
			\begin{align}\label{degNbdd}
				\deg_R(N_j):=\text{powers of $R$ in $N_j$}\in[j ,j+2k].
			\end{align}	
		\end{lemma}
		\begin{proof}
			We prove this by induction on $k=1,2,\ldots$. For the base case $k=1$, 	we use expansion \eqref{Rb-transf'}, which is of the form \eqref{RbnExp} with $N_1=Rc$ and  $N_0=h$, satisfying degree bound \eqref{degNbdd}.
			
			Suppose now \eqref{RbnExp} holds with some $k\ge1$, and we prove it for $k\to k+1$.
			Using \eqref{Rb-transf'} and the induction hypothesis, we write
			\begin{align}
				(Rb)^{k+1}=&(Rc+R_2h)(Rb)^k\notag 
				\\=& \sum_{l=0}^k R c R_2^l N_{k-l} + \sum_{l=0}^k R_2h  R_2^l N_{k-l}\notag
				\\=:&A+B.\label{Rbn1Exp}
			\end{align}
			The goal now is to commute the bounded operator $R_2$ successively to the left. Using the relation
			\begin{align}\label{R2form}
				R_2=s(1+zR)
			\end{align}
			and identity \eqref{R1h-comm}, we find
			\begin{align}\label{R2Bcomm}
				\Ad_{R_2}(D)=(s^{-1}R_2-1)\Ad_d(D)R , 
			\end{align}
			for any operator $D$ allowed by the domain consideration. Iterating identity \eqref{R2Bcomm} for $p\ge1$ times shows that there exist  absolute constants $c_0,\ldots, c_p$ s.t.
			\begin{align}\label{R2comml}
				\ad{p}{R_2}{D}=\sum_{q=0}^p c_qs^{-q}R_2^q \ad{p}{d}{D}R^p.
			\end{align}
			Moreover, for any bounded operators $D,\,E$ and integers $l\ge1$, we have
			\begin{align}
				\label{BAlExp}
				DE^l= E^lD+ \sum_{p=1}^l (-1)^p\binom{l}{p}E^{l-p}\ad{p}{E}{D}.
			\end{align}
			Applying \eqref{R2comml}--\eqref{BAlExp} to term $A$ in \eqref{Rbn1Exp} with $D=c$ and $E=R_2$,  and using that $[R_2,R]=0$, we find
			\begin{align}
				A\equiv &RcN_k+\sum_{l=1}^k R c R_2^l N_{k-l}\notag
				\\=&R c  N_{k } + \sum_{l=1}^kR_2^l R c  N_{k-l}  \notag
				\\&+\sum_{l=1}^k\sum_{p=1}^l (-1)^p\binom{l}{p}R_2^{l-p}R\ \ad{p}{R_2}{c}N_{k-l}\notag 
				\\=&  R c  N_{k } + \sum_{l=1}^kR_2^l R c  N_{k-l}   \notag 
				\\&+\sum_{l=1}^k\sum_{p=1}^l\sum_{q=0}^p (-1)^pc_qs^{-q}\binom{l}{p}R_2^{l-p+q} R\  \ad{p}{d}{c}R^pN_{k-l}.  \label{1552}
			\end{align}
			Regrouping \eqref{1552} according to the power in $R_2$ shows that
			\begin{align}
				A=&\sum_{l=0}^{k } R_2^l N^{(1)}_{k+1-l},\label{1561}
				\\N^{(1)}_{k+1}:=& RcN_k,\label{1562}
				\\N^{(1)}_{k+1-l}  =&  R c  N_{k-l}\label{1563}\\&+\sum_{l'= l}^k\sum_{\substack{p=1,\ldots,l'\\q=0,\ldots,p,\\
						q-p=l-l'}}
				(-1)^pc_qs^{-q}\binom{l'}{p}  R\  \ad{p}{d}{c}R^pN_{k-l'} , \quad l=1,\ldots, k .\notag
			\end{align}
			Since $\deg_R N_j\in [j, j+2k]$, we derive from expressions \eqref{1562}--\eqref{1563} that
			\begin{align}
				\deg_R({N^{(1)}_j} )\in  [j+1,j+2k+1], \quad j = 0,\ldots k.
			\end{align}
			
			Similarly, applying \eqref{R2comml}--\eqref{BAlExp} to term $B$ in \eqref{Rbn1Exp} with $D=h$ and $E=R_2$ yields
			\begin{align}
				B\equiv & R_2h N_{k}+ \sum_{l=1}^k R_2h  R_2^l N_{k-l}\notag
				\\=&R_2h  N_{k } + \sum_{l=1}^kR_2^{l+1} h  N_{k-l}  \notag
				\\&+\sum_{l=1}^k\sum_{p=1}^l (-1)^p\binom{l}{p}R_2^{l-p+1}\ \ad{p}{R_2}{h}N_{k-l}\notag 
				\\=&  R_2 h  N_{k } + \sum_{l=1}^kR_2^{l+1}h  N_{k-l}   \notag 
				\\&+\sum_{l=1}^k\sum_{p=1}^l\sum_{q=0}^p (-1)^pc_qs^{-q}\binom{l}{p}R_2^{l-p+q+1} \  \ad{p}{d}{h}R^pN_{k-l}.  \label{1553}
			\end{align} 
			Regrouping \eqref{1553} according to the power in $R_2$ shows that
			\begin{align}
				B=&\sum_{l=1}^{k+1 } R_2^l N^{(2)}_{k+1-l},\label{1566}
				\\N^{(2)}_{k}:=& hN_k,\label{1572}
				\\N^{(2)}_{k+1-l}  =& h  N_{k+1-l}\label{1568}\\&+\sum_{l'= l}^{k+1}\sum_{\substack{p=1,\ldots,l'\\q=0,\ldots,p,\\
						q-p=l-l'}}
				(-1)^pc_qs^{-q}\binom{l'-1}{p}   \  \ad{p}{d}{h}R^pN_{k+1-l'} ,  \notag
			\end{align}
			with $l=2,\ldots, k+1$ and		\begin{align}
				\deg_R({N^{(2)}_j} )\in  [j,j+2k+1], \quad j = 1,\ldots k+1.
			\end{align}
			Combining expansions \eqref{1561}, \eqref{1566} in line \eqref{Rbn1Exp} yields
			\begin{align}
				(Rb)^{k+1}= N^{(1)}_{k+1}+\sum_{l=1}^kR_2^l(N^{(1)}_{k+1-l}+N^{(2)}_{k+1-l}) + R_2^{k+1}N^{(2)} _{0},
			\end{align}which is of the form \eqref{RbnExp}  with $k\to k+1$. 
			This completes the induction and the proof of \lemref{lemRbnExp}. 
		\end{proof}
		
		Next, we have the following lemma
		\begin{lemma}\label{lemPhiRem1}
			Let $\Rem_1$ be as in \eqref{Ak-part-exp2}.   	If $K$ is any bounded operator with $\mathrm{ran} \, d \subset \ker K$ then 
			\begin{align}\label{chiRem1Est}
				K\Rem_1=O(\norm{K}).
			\end{align} 
		\end{lemma}

		\begin{proof} 
			We use expansion \eqref{RbnExp}. Since $\mathrm{ran} \, d \subset \ker K$,  we have  $K R_2=(K d) R=0$  by definition  \eqref{Rb-transf'}. Thus only the leading term in \eqref{RbnExp} survives left multiplication by $K$, yielding
			\begin{align}\label{pfchiRem1Est}
				K\Rem_1=&\int  K (Rb)^{n}R^E
				=\int K N_n R^E. 
			\end{align}
			By the definition of $N_n$ (see \lemref{lemRbnExp}),  we have 
			\begin{align}
				\norm{ N_n}\le C\sum _{j=n}^{3n}\abs{\Im(z)}^{-j}.\label{NnEst}
			\end{align} Thus, by \eqref{pfchiRem1Est},
			\begin{align}
				\norm{K\Rem_1}\le \norm{K}\sum _{j=n}^{3n}\int \abs{\Im(z)}^{-(j+1)}.
			\end{align}
			This, together with estimate \eqref{rempEst} with $(p,l)=(n,0),\ldots, (3n,0)$ (recall $n\ge1$ to begin with), implies the desired result, \eqref{chiRem1Est}.
		\end{proof}

		Applying \eqref{chiRem1Est} with $K= \chi_X^\sharp$, whose kernel contains $\mathrm{ran} \, d$ due to \eqref{dX}, we obtain
		\begin{align}\label {chiRem1Est'}
			\chi_X^\sharp\Rem_1=O(1).
		\end{align}
		Finally, plugging \eqref{chikEst} and  \eqref{chiRem1Est'} back to expansion \eqref{Ak-part-exp} yields the desired estimate \eqref{claim1'}.
		This completes the proof of \eqref{claim1'}. 
	\end{proof}

		\begin{remark}
			We mention the following alternative proof of Proposition \ref{cl:1}. Recalling that $\chi_{0s}=\chi(s^{-1}d_X^E)$ with $\chi$ supported on $[c_\delta,\infty)$ for some positive $c_\delta$, we write
			\begin{align*}
				\big\| \chi_{0s}\chi_X^\sharp\big\|=\big\|\chi_{0s}\big(d_X^E\big)^{-n}\big(d_X^E\big)^n\chi_X^\sharp\big\|\le(c_\delta s)^{-n}\big\|\big(d_X^E\big)^n\chi_X^\sharp\big\|.
			\end{align*}
			Now, with the convention $\prod_{i=2}^{n} A_i = A_2\dots A_{n}$, we have
			\begin{align*}
				\big(d_X^E\big)^n&=g(H)d_X\Big(\prod_{i=2}^{n} g^2(H)d_X\Big)g(H)\\
				&=g(H)d_X\langle x\rangle^{-1}\langle x\rangle\Big(\prod_{i=2}^{n}g^2(H)\langle x\rangle^{-i+1} d_X \langle x\rangle^{-1}\langle x\rangle^{i}\Big)g(H)\langle x\rangle^{-n}\langle x\rangle^n.
			\end{align*}
			A standard induction argument shows that $\langle x\rangle^ig^2(H)\langle x\rangle^{-i}$ is a bounded operator for any positive integer $i$ (since $H$ is the Schr\"odinger operator $H=-\Delta+V$), and likewise with $g$ instead of $g^2$. Since in addition $d_X \langle x\rangle^{-1}$ is bounded, we deduce that
			\begin{align*}
				\big\|\big(d_X^E\big)^n\chi_X^\sharp\big\|\le C_n\big\|\big \langle x\rangle^n\chi_X^\sharp\big\| \le C'_n,
			\end{align*}
			since $X$ is bounded. This establishes Proposition \ref{cl:1}. (Note that if $X$ is unbounded, the same holds, replacing $\langle x\rangle$ by $\langle d_X\rangle$ in the argument above.)
			
			The proof we gave {in \secref{secPfCl1}} has the advantage of being more robust. Moreover the arguments we used are also crucial in our proof of \eqref{claim2} given in the next section.
		\end{remark}
	
	\subsection{Proof of claim \eqref{claim2}}	\label{adap}

	Recall $\chi
	$, 
	$\tilde{g}$, and $\tilde \chi$ are smooth cutoff functions  such that $\supp(\tilde{g})\subset \{g= 1\}$ and  $\supp(\tilde{\chi}')\subset (\delta,+\infty)=\{\chi=1\}$ (see  Figs.~\ref{fig:tildechi}--\ref{fig:tildeg}).					Let $\bar{g}=1-g$ and $\bar{\chi}=1-\chi$.  It follows that 
	\begin{align}
		&\bar{g}(\mu)\tilde{g}(\mu)=0,\label{tildeg-cond}\\&\bar{\chi}(\mu)\tilde{\chi}(\mu)=0.\label{tildechi-cond}
	\end{align}
	
	In the remainder of this section, we use the following notations: For $s,\,v,\,t$ as in \eqref{chitsdef} and $z\in\Cb,\,\Im(z)\ne0$,
	\[
	\begin{aligned}
		d_t\equiv& d_X-vt,\qquad d^E_t\equiv d^E_X-vt=gd_Xg-vt,\\
		R\equiv& (d_t/s-z)^{-1},\qquad R^E\equiv(d^E_t/s-z)^{-1},
	\end{aligned}\]
	and 
	\begin{align}\label{xiDef}
		&\xi(\mu):=\sqrt{\chi(\mu)},\quad \bar{\xi}(\mu) =1-\xi(\mu), 
		\\	&\phi=\phi(d_t/s),\quad\phi^E=\phi(d^E_t/s) \text{ for }\phi\in \cX,
		\\&g=g(H),\quad \tilde g=\tilde g(H) \text{ for $g,\,\tilde g$ from \eqref{tildeg-cond}}.
	\end{align}
	Using these notations, we reproduce Claim \ref{claim2} as follows:
	\begin{proposition}\label{cl:2}	For every $\chi\in\cX$ and  $\tilde g,\,\tilde \chi$ as in \eqref{tildeg-cond}--\eqref{tildechi-cond},	\begin{align}\label{s2}
			\chi^E &\geq \tilde{g}\tilde{\chi} \tilde{g}+O(s^{-n}).
		\end{align}
	\end{proposition}

		\begin{proof}
			Since $\|\tilde g\tilde{\chi}\tilde g\|\leq 1$, we have
			\begin{align}\label{s2.2'}
				\chi^E&\geq \xi^E\tilde g\tilde{\chi}\tilde g\xi^E=\tilde g\tilde{\chi}\tilde g-\bar{\xi}^E\tilde g\tilde{\chi}\tilde g-\tilde g\tilde{\chi}\tilde g\bar{\xi}^E+\bar{\xi}^E\tilde g\tilde{\chi}\tilde g\bar{\xi}^E.
			\end{align}We now claim \begin{equation}\label{scl3'} 
				\bar{\xi}^E\tilde g\tilde{\chi}=O(s^{-n}). 
			\end{equation}
			If \eqref{scl3'} holds, then the last three terms on the r.h.s. of \eqref{s2.2'} are $O(s^{-n})$ and we are done.

			Since the operator $b\equiv d-d^E= d_t-d^E_t$ as in the proof of \propref{cl:1}, proceeding as in \eqref{Ak-part-exp}--\eqref{Ak-part-exp2}, we find the expansion
			\begin{align}\label{barxiExp1}
				\bar\xi ^E&=\sum_{k=0}^{n-1} \bar \xi_k+s^{-n}\Rem_1,
			\end{align}
			where
			\begin{align} \label{barxikDef}
				\bar\xi_k&=\int (Rb/s)^kR\,d\widetilde{\bar\xi}(z)\quad\text{and}\quad \Rem_1=\int (Rb)^{n}R^E\,d\widetilde{\bar\xi}(z), 
			\end{align}
			where  	$ \widetilde{\bar\xi}(z)$ is an almost analytic extension of the function $\bar\xi(\mu)$. (Below we will omit the measure $d\widetilde{\bar\xi}(z)$ when no confusion arises.) By expansion \eqref{barxiExp1}, Claim \eqref{scl3'} is equivalent to the relations
			\begin{align}
				&\bar \xi_k \tilde g\tilde \chi = O(s^{-n}),\label{scl3.1}
				\\&\Rem_1 \tilde g\tilde \chi = O(1). \label{scl3.2}
			\end{align}
			
			We first prove \eqref{scl3.1}.					We write the l.h.s. of \eqref{scl3.1} as
			\begin{align}\label{552}
				\bar\xi_k\tilde g\tilde{\chi}=\bar\xi_k\tilde{\chi}\tilde g+\bar\xi_k[\tilde g,\tilde{\chi}].
			\end{align}
			Since $\ad{k}{d_t/s}{\tilde g}=s^{-k}\ad{k}{d}{\tilde g}$ is bounded for $0\le k \le n$, we have by expansion  \eqref{4.Hcomm-exp-right} that
			\begin{align}\label{592}
				[\tilde g,\tilde \chi]=\sum_{k=1}^{n-1}(-1)^k\frac{s^{-k}}{k!}\tilde\chi^{(k)}(d_t/s)\ad{k}{d}{\tilde g}+s^{-n}\Rem_3,
			\end{align}			
			where $\Rem_3=O(1)$. Plugging \eqref{592} into \eqref{552} yields
			\begin{align}
				\notag\bar\xi_k\tilde g\tilde{\chi}=&\bar\xi_k\tilde{\chi}\tilde g+ \sum_{k=1}^{n-1}\frac{s^{-k}}{k!}\bar\xi_k\tilde\chi^{(k)}(d_t/s)\ad{k}{d}{\tilde g}+s^{-n}\bar\xi_k\Rem_3
				\\=:&A+B+C.\label{554}
			\end{align}
			We apply \corref{corChik} to the function $\bar\xi $ to obtain the expansion
			\begin{align}\label{barxiExp2}
				\bar{\xi}_k =&h^k\bar\xi(d_t/s)+\sum_{q=0}^k\sum_{l=0}^{n-1}s^{-l} B_{q,l}(\bar{\xi} p_{q,l})^{(l+1)}(d_t/s)+s^{-n}\Rem_{2,k},
			\end{align}
			where $\norm{h}\le 2$, $B_{q,l}=O(1)$  are  defined in \lemref{lem:(Rb/s)k-exp}, part (b), and $\Rem_{2,k}=O(1)$. 
			Thus $\bar\xi_k=O(1)$ and so the term $C$ in line \eqref{554} is $O(s^{-n})$. 
			By definition \eqref{xiDef}, we have \begin{align}\label{xichirel}
				\bar\xi^{(l)}(\mu)\tilde\chi^{(m)}(\mu)=0\quad \text{ for any integers } l,m\ge0,
			\end{align}see \figref{fig:tildechi} .
			Thus, inserting  \eqref{barxiExp2} to \eqref{554} and using \eqref{xichirel}, we find
			\begin{align}
				A=&s^{-n}{\sum_{k=0}^{n-1}\Rem_{2,k}} \tilde\chi \tilde g
				=O(s^{-n}),\label{AEst1}
				\\	B=&s^{-n} \sum_{k=1}^{n-1}\sum_{l=0}^{n-1}\frac{s^{-k}}{k!}{\Rem_{2,l}}\tilde\chi^{(k)} \ad{k}{d}{\tilde g}
				=O(s^{-n}).\label{BEst1}
			\end{align}
			Thus we have proved \eqref{scl3.1}.
			
			Next, we prove \eqref{scl3.2} by the following lemma:
			\begin{lemma}\label{lemform5}
				For   $k=1, \ldots, n$ and $\Rem_1(k):=	\int (Rb)^kR^E$, 
				\begin{align}\label{form5}
					\Rem_1(k)\tilde g\tilde \chi= O(1).
				\end{align}

			\end{lemma}
			\begin{proof}

				We prove this by induction on $k$. We have by expansion \eqref{Rb-transf'} that $Rb=Rc  +R_2h$. For the base case $k=1$, we write
				\begin{align}\label{R1tg}
					RbR^E  =&  RcR^E +R_2 R^Eh +  R_2  [h, R^E]  \notag\\
					=&RcR^E +R_2R^Eh + s^{-1}R_2R^E\ad{}{d^E}{h}R^E,
				\end{align}	
				where we use the relation (c.f. \eqref{R1h-comm})
				\begin{align}\label{REcomm} 
					[B, R^E]= s^{-1}R^{E} \Ad_{d^E}(B)R^{E},
				\end{align}	
				valid for any operator $B$ allowed by the domain consideration. 
				
				The second term \eqref{R1tg} is a priori large $O(s)$ but it is removed by $\tilde g.$ Indeed, since $h\tilde g=0$ by \eqref{h} and the relation \eqref{tildechi-cond} (c.f. \figref{fig:tildeg}), and $s^{-1}R_2=1+zR$ by \eqref{R2form}, we have 
				\begin{align}
					\Rem_1(1)\,\tilde g =&  \int RcR^E\tilde g+\int s^{-1}R_2R^E\ad{}{d^E}{h}R^E\tilde g\notag
					\\	=& \int RcR^E\tilde g+\int R^E\ad{}{d^E}{h}R^E\tilde g + \int zRR^E\ad{}{d^E}{h}R^E\tilde g. \label{5101}
				\end{align}
				For $f\in C_c^\infty(\Rb)$, the operators $\ad{k}{d^E}{f}$ are $O(1)$	by results from \secref{sec:mult-comm-est}, see \eqref{3.15'} and \cite[eqn.~(B.20)]{HunSig1}. Thus the three integrals in line \eqref{5101} are $O(1)$ by the estimates $\norm{\tilde g}\le 1$, $\|c\|, \norm{\Ad_{d^E}(h)}=O(1)$, $\norm{R},\,\norm{R^E}\le \abs{\Im(z)}^{-1}$, and \corref{corRempEst} with $(p,l)=(1,0),\,(2,1)$. 
				This shows \eqref{form5} with $k=1$.

				Suppose now \eqref{form5} holds with some $k\ge1$, and we prove it for $k\to k+1$.
				First, we note the relation $R^E-R=RbR^E$ and so
				\begin{align}\label{Rem1kform}
					\Rem_1(k )=&\int (Rb)^kR + \int (Rb)^k(R^E-R)\notag
					\\=& \int (Rb)^kR + \int (Rb)^{k+1}R^E
					=s^k\bar\xi_k+\Rem_1(k+1), 
				\end{align}
				where $\bar\xi_k$ is defined by \eqref{barxikDef}. Right-multiplying $\tilde g\tilde \chi$ on both sides of \eqref{Rem1kform} and rearranging, we find 
				\begin{align}
					\Rem_1(k+1)\tilde g\tilde \chi = \Rem_1(k)\tilde g\tilde \chi - s^{k}\bar \xi_k\tilde g\tilde \chi. 
				\end{align}
				The first term on the r.h.s. is $O(1)$ by induction hypothesis. The second term is $O(s^{k-n})$ by \eqref{scl3.1} proved earlier.  Since $k\le n$, this completes the induction and the proof of \lemref{lemform5}.
								\end{proof}
				
				Since $\Rem_1\equiv\Rem_1(n)$ in \lemref{lemform5}, estimate \eqref{form5} implies \eqref{scl3.2}. This, together with \eqref {scl3.1}, implies the claim \eqref{scl3'}.  This completes the proof of \propref{cl:2}.    \end{proof}

			\appendix

	\section*{Acknowledgment}
	The research of M.L.\ is supported in part by the DFG under grant SFB TRR 352.
 The research of  D.O.\ and I.M.S.\ is supported in part by NSERC Grant NA7901. 		
	J. Z.\ is supported by DNRF Grant CPH-GEOTOP-DNRF151, DAHES Fellowship Grant 2076-00006B, DFF Grant 7027-00110B, and the Carlsberg Foundation Grant CF21-0680. His research was also supported in part by NSERC Grant NA7901.  
	Parts of this work were done while he was visiting MIT.

	\section*{Declarations}
	\begin{itemize}
		\item Conflict of interest: The Authors have no conflicts of interest to declare that are relevant to the content of this article.
		\item Data availability: Data sharing is not applicable to this article as no datasets were generated or analysed during the current study.
	\end{itemize}
	
	\appendix


					\section{Existence of unique 
	solution to vNL equation} \label{sec:exist}
In  this section, we prove existence of unique mild solution to \eqref{vNLeq} in the Schatten space $S^1$ of trace-class operators. 
 Throughout the section, we assume \ref{W1}, i.e. $\sum\nolimits_{j\geq 1}W_{j}^{*}W_{j}$ with $W_j$ in \eqref{vNLeq} converges strongly. 

The main mechanism is the following theorem (see e.g. \cite[Theorem 3.1.33]{BR1}):
\begin{theorem}\label{thm:exist-sol}
	Let $U$ be a strongly continuous semigroup on the Banach space $X$ with generator $S$ and let $P$ be a bounded operator on $X$.  Then, $S+P$ generates a strongly continuous semigroup $U^P$.
\end{theorem}
In our case, $X$ is the Schatten space $\cS^1$ with trace-norm $\|\cdot\|_{1}$, the strongly continuous semigroup $U$ is the unitary semigroup generated by $-i[H,\cdot]$ and the perturbation $P$ is the Lindblad operator $G$ (see \eqref{vNLeq}).

In the next lemma, we show that $G$ is norm closed and bounded, so that \thmref{thm:exist-sol} indeed applies.

\begin{lemma}\label{lem:G-bdd-cl}
	The Lindblad operator $G$ defined in \eqref{vNLeq} is bounded on $\cS_1$.
\end{lemma}
\begin{proof}
	
	Without loss of generality, we assume $\rho\in\cS^1$ is positive.  Let $G_j(\cdot)=W_j\del{\cdot} W_j^*-\tfrac{1}{2}\{W_j^*W_j,\del{\cdot}\}$.  For a positive $\rho$, it is clear the operators $W_j\rho W_j^*$ and $\{W_j^*W_j,\rho\}$ are positive for all $j$.  Then, by cyclicity of the trace, we have
	\begin{align}\label{Gj'}
		\|G_j(\rho)\|_1&\leq \|W_j\rho W_j^{*}\|_1+\frac{1}{2}\|\{W_j^*W_j,\rho\}\|_1\nonumber\\
		&\leq \Tr|W_j\rho W_j^*|+\frac{1}{2}\Tr|\{W_j^*W_j,\rho\}|\nonumber\\
		&=\Tr(W_j\rho W_j^*)+\frac{1}{2}\Tr(\{W_j^*W_j,\rho\})\nonumber\\
		&=2\Tr(W_j^*W_j\rho).
	\end{align}
	Thus, 
	\begin{align}
		\|G(\rho)\|_1&=\norm{\sum_{j\geq 1}G_j(\rho)}_1\leq 2\sum_{j\geq 1}\Tr(W_j^*W_j\rho)\leq 2\norm{\sum_{j\geq 1}W_j^*W_j}\|\rho\|_1.
	\end{align}
	Since $\sum\nolimits_{j\geq 1}W_j^*W_j$ is bounded by the uniform boundedness theorem, this proves $G$ is bounded on $\cS_1$, which completes the proof.
\end{proof}
Theorem \ref{thm:exist-sol} shows that \eqref{vNLeq} has a unique strong solution in $\cD(L)$ and a unique mild solution in $\cS_1$.  We denote the semigroup generated by vNL operator $L$ by $\beta_t$ as before.
Note that  since $e^{L_0 t}$ is a group (defined on $\R$), then so is $\beta_t=e^{L t}$.

\medskip

The positivity preserving property of $\beta_t$ follows from		 \cite[Theorem 5.2]{Davies}.  We summarize the key result in the following lemma:
\begin{lemma}\label{lem:pos-vNL}
	The semigroup $\beta_t$ is positive for all $t\geq 0$.
\end{lemma}
\begin{proof}
	First, we rewrite the vNL operator $L$ as
	\begin{align}
		L(\rho)&=-iK_{H+iP}(\rho)+F(\rho),
	\end{align}
	where $P=P^*=\tfrac{1}{2}\sum\nolimits_{j\geq 1}W_j^*W_j$, $K_A(\rho)=A\rho-\rho A^*$ and $F(\rho)=\sum\nolimits_{j\geq 1}W_j\rho W_j^*$.  
	
	Let $B_t=e^{-iHt-Pt}$, which is well-defined since $P$ is bounded by assumption.  It is easy to check that the semigroup $S_t$ generated by $-iK_{H+iP}$ is given by
	\begin{align}\label{B-semi}
		S_t(\rho)&=B_t\rho B_t^{*},
	\end{align}
	which obviously defines a positive semigroup.  On the other hand, since $$\sum_{j\geq 1}W_j\rho W_j^*\geq 0$$ for all $\rho\geq 0$, then $F$ generates a positive semigroup $e^{Ft}$.  
	
	Finally, by Trotter-Lie formula, we have
	\begin{align}\label{Davies-def}
		\beta_t(\rho)&=\lim_{n\rightarrow\infty}(S_{t/n}e^{Ft/n})^n(\rho),
	\end{align}
	where the limit is taken in the trace-norm. Hence the semigroup $\beta_t$ is positive.
\end{proof}

Note that \eqref{Davies-def} yields another way to construct the semigroup $\beta_t=e^{Lt}$.

					\section{Remainder estimates}\label{secRemEst}
In this appendix and the next one, we present some estimates and commutator expansions, first derived in \cite{SigSof}  and then improved in \cite{GoJe,HunSig1, HunSigSof, Skib}.  We adapt some of the arguments from \cite{HunSig1} and refer to this paper for details and references.

Throughout this section we fix an integer $\nu\ge0$. For integers $p\ge 0$ 
and smooth functions $f\in C^{\nu+2}(\Rb)$, we define a weighted norm 
\begin{align}\label{pnNorm}
	\cN(f,p):=\sum_{m=0}^{\nu+2}\int_\Rb \br{x}^{m-p-1}\abs{f^{(m)}(x)}\,dx.
\end{align}

Note that
\begin{align}\label{Norder}
	p\le p'\implies \cN(f,p')\le \cN(f,p),
\end{align}
and we have the following property:
\begin{lemma}\label{lemNfin}
	Let $p\ge0$ be an integer. Suppose $f\in C^{\nu+2}$ and there exist  $C_0,\,\rho>0$ such that, for $ m=0,\ldots, \nu+2$, 
	\begin{align}
		\label{fCond}
		\norm{	\br{x}^{m-p+\rho}	f^{(m)}(x)}_{L^\infty}\le C_0 .
	\end{align}  Then  there exists $C>0$ depending only on $\rho,\,C_0,\,\nu$ such that
	\begin{align}\label{Ncond}
		\cN(f,p)\le C.
	\end{align}
\end{lemma}
\begin{proof}
	We have 
	\begin{align*}
		\cN(f,p) \le&  \sum_{m=0}^{\nu+2}	\norm{	\br{x}^{m-p+\rho}	f^{(m)}(x)} \int_\Rb \br{x}^{-1-\rho}dx
		\\\le& (\nu+3) C_0 \int_\Rb \br{x}^{-1-\rho}dx,
	\end{align*}
	and the integral converges for $\rho>0$. 
\end{proof}
\begin{corollary}\label{lemNfin'} 
	Let $p$ and $l$ be two integers with  $ p > l \ge 0$.
	If $f\in C^{\infty}(\Rb)$ and $f^{(l+1)} $ {has compact support}, 
	then  \eqref{Ncond} holds.
\end{corollary}
\begin{proof}
	It suffices to verify condition \eqref{fCond} for the function $f$, whence \eqref{Ncond} follows from Lemma~\ref{lemNfin}.
	For $m\ge l+1$,   \eqref{fCond} holds since $f^{(m)}\in C_c^\infty$. For $m\le l$,   integrating $f^{(l+1)}$  shows that $\abs{f^{(m)}(x)}\le C \br{x}^{l-m}$. Since $p\ge l+1$, we have	\eqref{fCond} with $\rho=1$.  \end{proof}

Write $z=x+iy\in\Cb$. In what follows, 					as in \cite[Eq. (B.5)]{HunSig1}, for $f\in C^{\nu+2}(\Rb)$, we take $\tilde f(z)$ to be an almost analytic extension of $f$ defined by 	
\begin{equation}\label{tfDef}
	\tilde f (z):=\eta\del{\frac{y}{\br{x}}}\sum_{k=0}^{\nu+1}f^{(k)}(x)\frac{(iy)^k}{k!},
\end{equation}							where $\eta\in C_c^\infty(\Rb)$ is a cutoff function   with  	$\eta(\mu)\equiv1$ for $\abs{\mu}\le1$,  $\eta(\mu)\equiv0$ for $\abs{\mu}\ge2$, and $\abs{\eta'(\mu)}\le1$ for all $\mu$.
This $\tilde f(z)$ induces a measure on $\Cb$ as
\begin{align}\label{measDef}
	d\tilde f(z):=-\frac{1}{2\pi}\di_{\bar z}\tilde f(z)dx\,dy.
\end{align}
In the remainder of this appendix, we derive integral estimate for various functions against the measure \eqref{measDef}.

The next result is obtained by adapting the argument  in \cite[Lem.~B.1]{HunSig1}:
\begin{lemma}[Remainder estimate]\label{lemRemEst}
	Let $0\le p \le \nu$.
	Let $f\in C^{\nu+2}(\Rb)$ satisfy \eqref{Ncond}.
	Then the extension $\tilde f$ from \eqref{tfDef} satisfies the following estimate for some 						  $C=C(f,\nu,p)>0:$
	\begin{align}\label{fRemEst}
		\int \abs{d\widetilde{f}(z)} \abs{\Im(z)}^{-(p+1)} \le C.
	\end{align}
\end{lemma}
\begin{proof}

	Differentiating formula \eqref{tfDef}, we obtain the estimate
	\begin{align}\label{dzfEst}
		\abs{\di_{\bar z} \tilde f(z)}\le
		\eta\del{\frac{y}{\br{x}}}\frac{\abs{y}^{\nu+1}}{(\nu+1)!}\abs{f^{(\nu+2)}(x)}+\sum_{k=0}^{\nu+1}\rho\del{\frac{y}{\br x}}\frac{\abs{y}^k}{k!}\abs{\frac{1}{\br{x}}f^{(k)}(x)},
	\end{align}
	where \begin{align}\label{rhodef}
		\rho(\mu):=\abs{\eta'(\mu)}\br{\mu}
	\end{align} is supported on $1<\abs{\mu}<2$. 
	
	For each fixed $x$, we define
	\begin{align}\label{Gxdef}
		G(x):=p.v.\int \abs{\di_{\bar z}f(z)}\abs{y}^{-(p+1)}\,dy
	\end{align}
	by integrating \eqref{dzfEst} against $\abs{y}^{-(p+1)}$.  Using that 	$\eta(y/\br x)\equiv 0$ for $\abs{y}>\br{x}$ and $\rho(y/\br{x})\equiv 0$ for $\abs{y}\le \br{x}$ or $\abs{y}\ge2\br{x}$, we find
	\begin{align}
		G(x)\le&\int_{\abs{y}\le \br{x}} \frac{\abs{y}^{\nu-p}}{(\nu+1)!}\eta\del{\frac{y}{\br{x}}}\,dy\abs{f^{(\nu+2)}(x)}\label{549}
		\\&+\sum_{k=0}^{\nu+1}\int_{\br{x}<\abs{y}<2 \br{x}} \rho\del{\frac{y}{\br x}}\frac{\abs{y}^{k-p-1}}{k!}\,dy\abs{\frac{1}{\br{x}}f^{(k)}(x)}.\label{549'}
	\end{align}
	Since $0\le \eta(\mu)\le1$ and $\nu\ge p$, the integral in line \eqref{549} converges and can be bounded as
	\begin{align}
		\int_{\abs{y}\le \br{x}} \frac{\abs{y}^{\nu-p}}{(p+1)!}\eta\del{\frac{y}{\br{x}}}\,dy\abs{f^{(p+2)}(x)}\le\frac{2\br{x}^{\nu-p+1}}{(p+1)!}\abs{f^{(p+2)}(x)}.\label{RemEst1}
	\end{align}
	
	To bound line \eqref{549'}, we use that $\rho(y/\br{x})< \sqrt{5}$ and $\abs{y}^{k-p-1}\le \br{x}^{k-p-1}$ for $\br{x}<\abs{y}<2\br{x}$,  $0\le k\le p+1$ (see \eqref{rhodef}). Thus  each integral in line \eqref{549'} can be bounded as
	\begin{align}
		&\sum_{k=0}^{\nu+1}\int_{\br{x}<\abs{y}<2 \br{x}} \rho\del{\frac{y}{\br x}}\frac{\abs{y}^{k-p-1}}{k!}\,dy\abs{\frac{1}{\br{x}}f^{(k)}(x)}\notag\\\le&{\sum_{k=0}^{p+1}  \frac{4\sqrt{5} \br{x}^{k-p-1}}{k!}\abs{f^{(k)}(x)}
			+\sum_{k=p+1}^{\nu+1}  \frac{\sqrt{5}\cdot 2^{k-p+1}\br{x}^{k-p-1}}{k!}\abs{f^{(k)}(x)}}
		.\label{RemEst2}
	\end{align}
	Combining \eqref{RemEst1}--\eqref{RemEst2} in \eqref{549'}, we conclude that  
	\begin{align}
		\abs{G(x)}\le C F(x),\quad F(x):= \sum_{m=0}^{\nu+2}  \br{x}^{m-p-1}\abs{f^{(m)}(x)}.
	\end{align}

	Let $G_\lambda(x):=\one_{[-\lambda,\lambda]} G(x)$ with $\lambda>0$. Then   $G_\lambda\in L^1$ and $\abs{G_\lambda(x)}\le CF(x)$ for any $\lambda$. By assumption \eqref{Ncond} and definition\eqref{pnNorm}, we have $\norm{F}_{L^1}=\cN(f,p)<\infty$ and so $F\in L^1$. Therefore, sending $\lambda\to\infty$ and using the dominated convergence theorem   yields $G\in L^1$ with \begin{align}\label{GFest}
		\norm{G}_{L^1}\le C\norm{F}_{L^1}.
	\end{align} Recalling definition \eqref{Gxdef}, we find $(2\pi)^{-1}\norm{G}_{L^1}=$l.h.s. of \eqref{fRemEst}. 
	Thus we conclude   \eqref{fRemEst} from \eqref{GFest}.
\end{proof}

\lemref{lemRemEst} and	\corref{lemNfin'} together imply the following results:
\begin{corollary}\label{corRempEst'} 
	Let $p$ and $l$ be two integers with  $\nu \ge p > l\ge0$.
	If $f\in C^{\infty}(\Rb)$ and $f^{(l+1)} $ {has compact support}, 
	then  there exists $C>0$ such that the extension $\tilde f$ from \eqref{tfDef} satisfies the remainder estimate \eqref{fRemEst}.
\end{corollary} 

\begin{corollary}\label{corRempEst} 
	Let $p$ and $l$ be two integers with  $\nu \ge p > l \ge 0$.
	Let $P_l(x)$ be a polynomial   with $\deg \le l$. Let $f\in C^{\infty}(\Rb)$ have compactly supported derivatives. 
	Then there exists $C>0$ such that the extension $\tilde f$ from \eqref{tfDef} satisfies
	\begin{align}\label{rempEst}
		\int \abs{d\widetilde f (z)P_l(z)} \abs{\Im(z)}^{-(p+1)} \le C.
	\end{align}
\end{corollary}

\begin{proof}
	
	Let  $ f_l(x):=P_l(x)\chi(x)$.   Observe that  since $\di_{\bar z}P_l(z)=0$, we have by \eqref{measDef} that
	\begin{align}\label{measForm}
		P_l(z) d\widetilde f(z)= \,d\widetilde {f_l}(z).
	\end{align}
	We compute
	\begin{align}
		f_l^{(l+1)}=P_l^{(l+1)}f+ \sum_{k=0}^l \binom{l+1}{k}P_l^{(k)}f^{(l+1-k)}.
	\end{align}
	The term leading term on the r.h.s. vanishes since $\deg p\le l$. Each term in the sum lies in $C_c^\infty$ since $f^{(q)}\in C_c^\infty $ for $q\ge1$. Thus $f_l$ verifies the condition of \corref{corRempEst'} and so \eqref{rempEst} follows.
\end{proof}

					\section{Commutator expansions}\label{4.sec:commut}

In this appendix, we take 
$\tilde f(z)$,  $d\widetilde f(z)$ to be as in \eqref{tfDef}--\eqref{measDef}.

We frequently use the following result, taken from \cite[Lemma B.2]{HunSig1}:
\begin{lemma}
	\label{lemHSj-rep}
	Let $f\in C^{\nu+2}(\Rb)$ satisfy \eqref{Ncond} for some $p\ge0$. 
	Then for any self-adjoint operator $A$ on $\cH$, 
	\begin{align}
		\label{HSj-rep}
		\frac{1}{p!}f^{(p)}(A)=\int_\Cb d\tilde f(z)(z-A)^{-(p+1)},
	\end{align}
	where the integral   converges absolutely in operator norm and is uniformly bounded  in $A$.  
\end{lemma}
\begin{remark}
	Note that \eqref{Ncond} ensures $f^{(p)}$ is bounded independent of $A$ and  the remainder estimate in \lemref{lemRemEst} ensures the norm convergence of the r.h.s. of \eqref{HSj-rep}.
\end{remark}

We call Equation \eqref{HSj-rep} the \textit{Helffer-Sj\"ostrand (HS) representation}. 
	It is possible to obtain stronger results with less regularity assumption on $f$ using some technical estimates from \cite[Sec.~5]{ABG}. We do not pursue this generality here, as the assumption \eqref{Ncond} already suffices for our purposes. 

The HS representation \eqref{HSj-rep}, together with the remainder estimate \eqref{fRemEst}, implies the following commutator expansion:

\begin{lemma}\label{lemA.2} 
	Let $n\ge1$.
	Let $f\in C^{n+3}(\Rb)$ satisfy \eqref{Ncond} with $p=1$. 
	Let $A$ be an operator on $\cH$. Let $\Phi$ be a lower semi-bounded self-adjoint operator on $\cH$. {Let $f_s:=f(s^{-1}(\Phi-\al))$ for some fixed $\al$ and all $s>0$.}
	Suppose there exists $c\ge0$ such that
	\begin{equation}\label{phi-dom-cond}
		(\Phi+c)^{-1}\cD(A)\subset \cD(A),
	\end{equation}
	and 
	\begin{equation}\label{A.4}
		B_k:=		\ad{k}{\Phi}{A}\in\cB(\cH)\quad (1\le k\le n+1).
	\end{equation}
	Then $[A, f_s]\in\cB(\cH)$, and we have the expansion
	\				\begin{align}\label{4.Hcomm-exp} [A, f_s]= &- \sum_{k=1}^n{s^{- k}\over{k!}}B_kf^{(k)}_s -s^{-(n+1)}\Rem_{\rm left}(s)\\
		=&
		\sum_{k=1}^n(-1)^k{s^{- k}\over{k!}}f^{(k)}_s B_k +(-1)^{n+1}s^{-(n+1)}\Rem_{\rm right}(s),\label{4.Hcomm-exp-right} 
	\end{align}
	where 
	the remainders are defined by these relations and given explicitly by \eqref{left-rem}--\eqref{right-rem}.
	Moreover,  	 there exists $c>0$ depending only on $n$ and  $\cN(f,n+1)$, such that
	\begin{align}
		\norm{\Rem_{\rm left}(s)}_{\rm op}+\norm{\Rem_{\rm right}(s)}_{\rm op}\le& c\norm{B_{n+1}}.\label{4.A.21}
	\end{align}

\end{lemma} 
\begin{proof}
	Within this proof we write $R=(z-x_s)^{-1}$ with  $x_s=s^{-1}(\Phi-\al)$.			
	Hypothesis \eqref{phi-dom-cond} shows that  $$R=(\Phi+c)^{-1}(z (\Phi+c)^{-1}-x_s(\Phi+c)^{-1})^{-1}$$ maps $\mathcal{D}(A)$ into itself for $z$ with large $|\Im(z)|$ and therefore for all $z$ with $\Im(z)\neq 0$.  			

	It follows that 
	\begin{equation}\label{4.A.26}
		\big[ A , R \big ]=-s^{-1}R\Ad_\Phi(A)R
	\end{equation}
	holds in the sense of quadratic forms on $\mathcal{D}(A)$. Since $R$ is bounded and $\Ad_\Phi(A)$ is bounded by assumption, the r.h.s. of \eqref{4.A.26} is bounded and so $[A,R]$ extends to a bounded operator on $\cH$.
	
	Using \eqref{4.A.26}, we   proceed by commuting successively the commutators $B_k:=\ad{k}{\Phi}{A}$ to left and right, respectively. This way we  obtain
	\begin{align}
		&[A , R]\notag
		\\=&-\sum_{k=1}^n s^{-k}B_k R^{k+1} -s^{-(n+1)}RB_{n+1}R^{n+1} \label{B25}
		\\=&\sum_{k=1}^n (-1)^ks^{-k}R^{k+1}B_k +(-1)^{n+1}s^{-(n+1)}R^{n+1}B_{n+1}R,\label{B26}
	\end{align}					which hold on all of $\cH$ since  $B_k$'s are bounded operators by assumption \eqref{A.4}.

	Since $f$ may not decay at $\infty$, we cannot directly express $f_s=f(s^{-1}(\Phi-\al))$ using the HS representation \ref{HSj-rep}. We therefore introduce a cutoff as follows. Let $\eta^\lambda\in C_c^\infty(\Rb)$, $\lambda>0$ be cutoff functions   with  	$\eta^\lambda(x)\equiv1$ for $\abs{x}\le\lambda$,  $\eta(x)\equiv0$ for $\abs{\mu}\ge\lambda+1$, and $\norm{\eta^\lambda}_{C^{n+3}}\le C$ for all $\lambda$.
	Set $f^\lambda:= \eta^\lambda f$.
	Since $f^\lambda\in C_c^{n+3}$, it satisfies \eqref{Ncond} for all $p\ge0$. Thus the HS representation  \ref{HSj-rep} holds with $p=0$ and so 
	\begin{align}\label{4.HSj-rep7}
		[A, f_s^\lambda]=\int d\widetilde {f^\lambda}(z) \big[ A , R \big ],
	\end{align}
	which holds a priori on $\cD(A)$. 
	
	Plugging expansions \eqref{B25}--\eqref{B26} into \eqref{4.HSj-rep7} yields
	\begin{align}
		&\quad[A, f_s^\lambda]\notag\\&=-\sum_{k=1}^n{s^{- k}\over{k!}} B_k\int d\widetilde {f^\lambda}(z) R^{k+1} -s^{-(n+1)}\Rem_{\rm left}^\lambda(s),\label{4.HSj-rep9}\\
		&=\sum_{k=1}^n(-1)^k{s^{- k}\over{k!}}\int d\widetilde {f^\lambda}(z) R^{k+1} B_k+(-1)^{n+1}s^{-(n+1)}	\Rem_{\rm right}^\lambda(s)\label{4.A.20},
	\end{align}
	where
	\begin{align}
		\Rem_{\rm left}^\lambda(s)&=\int d\widetilde {f^\lambda}(z)RB_{n+1}R^{({n+1})},\label{left-rem}\\
		\Rem_{\rm right}^\lambda(s)&=\int d\widetilde {f^\lambda}(z)R^{({n+1})}B_{n+1}R\label{right-rem}. 
	\end{align}    
	Since the operator $B_{n+1}$ is bounded independent of $\lambda,\,z$, and $\norm{R}\le \abs{\Im(z)}^{-1}$, we have
	\begin{align}\label{4.A.23}
		\notag  &\norm{\Rem_{\rm left}^\lambda(s)}_{\rm op}+\norm{\Rem_{\rm right}^\lambda(s)}_{\rm op}\\\le& 2 \| B_{n+1}\|  \int |d\widetilde {f^\lambda}(z)| R^{n+2}\notag\\
		\le &2 \| B_{n+1}\|  \int |d\widetilde {f^\lambda}(z)| |\Im(z)|^{-(n+2)}.
	\end{align}
	Similarly we could bound the sums in \eqref{4.HSj-rep9}--\eqref{4.A.20}. Thus we see $[A,f^\lambda_s]$ extends to a bounded operator on $\cH$ for each $\lambda$.
	
	By \eqref{Norder} and the assumption $\cN(f,1)\le C$, $f$ satisfies condition \eqref{Ncond} with $p=1,\ldots, n+1$. Hence,  sending $\lambda\to \infty$ in \eqref{4.HSj-rep9}--\eqref{right-rem} and using \eqref{HSj-rep} for $p=1,\ldots, n$ and the remainder estimate \eqref{fRemEst} for $p=n+1$, we conclude that $[A,f_s]\in \cB(\cH)$ and  expansions 
	\eqref{4.Hcomm-exp}--\eqref{4.Hcomm-exp-right} and  estimate   \eqref{4.A.21} hold.\end{proof}    

	\bibliographystyle{plainnat}

\begin{thebibliography}{1}
	
	
	
	%
	%
	%
	\bibitem{AlickiLendi}
	R.~Alicki and K.~Lendi.
	\textit{Quantum dynamical semigroups and applications}, volume 286 of
	{Lecture Notes in Physics}.
	\newblock Springer-Verlag, Berlin, 1987.
	
	
	
	\bibitem{ABG} W.~Amrein, A.~Boutet de Monvel, V.~Georgescu, \textit{$C_0$-groups, commutator methods and spectral theory of $N$-body Hamiltonians}, Progress in Mathematics, Birkh\"{a}user Verlag, Basel (1996)
	
	\bibitem{AFPS} J. Arbunich, J. Faupin, F. Pusateri, I.M.~Sigal, \textit{Maximal Speed of Quantum Propagation for the Hartree equation}, To appear in Communications in Partial Differential Equations
	
	\bibitem{APSS}  J. Arbunich, F. Pusateri, I.M.~Sigal, A. Soffer, \textit{Maximal speed of quantum propagation}. Lett.\ Math. Phys. \textbf{111} (2021), No. 62. 
	
	
	
	
	%
	%
	%
	%
	%
	%
	%
	%
	%
	%
	%
	%
	\bibitem{BFS}  J.-F. Bony,  J. Faupin,  I.M. Sigal,
	\textit{Maximal velocity of photons in non-relativistic QED}, Adv. Math. \textbf{231} (2012), 3054--3078.
	%
	%
	%
	
	
	
	\bibitem{BR1}
	O.~Brattelli, D.W.~Robinson, \textit{Operator Algebras and Quantum Statistical Mechanics 1, 2nd Edition}, Springer Berlin, Heidelberg.  
	
	\bibitem{BFLS}
	S.~Breteaux, J.~Faupin, M.~Lemm, and I.~M.~Sigal, \textit{Maximal Propagation Speed in Open Quantum Systems},   In: The Physics and Mathematics of Elliott Lieb, The 90th Anniversary Volume I, Rupert L. Frank et al. (eds.), 109–130, EMS Press, Berlin, 2022
	
	
	%
	%
	%
	%
	
	
	
	\bibitem{CFKS} H. Cycon, R. Froese, W. Kirsch and B. Simon, \textit{
		Schr\"odinger operators}, Texts and Monographs in Physics, Springer Verlag (1987).
	
	
	%
	%
	%
	%
	\bibitem{Davies} E.~B.~Davies, \textit{Quantum Theory of Open Systems}, Academic Press 
	1976.
	%
	%
	%
	%
	%
	%
	%
	%
	%
	%
	\bibitem{Der} J. Derezi\'nski: 
	\textit{Asymptotic completeness of long-range $N$-body quantum systems},
	Ann. of Math. {\bf 138}, 427-476 (1993)
	%
	%
	%
	%
	%
	%
	%
	%
	%
	%
	%
	%
	%
	%
	
	
	
	\bibitem{EW}
	J.~Epstein, and K.B.~Whaley, \textit{Quantum speed limits for quantum-information-processing tasks} Phys. Rev. A, \textbf{95},(4) (2017), 042314
	
	\bibitem{FFFS} M. Falconi, J. Faupin, J. Fr\"ohlich, and B. Schubnel. \textit{Scattering theory for Lindblad master equations.} Commun. Math. Phys., \textbf{350}(3):1185--1218, 2017.
	%
	%
	%
	%
	\bibitem{FLS1}J.~Faupin, M.~Lemm, and I.~M.~Sigal, \textit{Maximal speed for macroscopic particle transport in the Bose-Hubbard model}, Phys. Rev. Lett. \textbf{128} (2022), 150602
	%
	\bibitem{FLS2}J.~Faupin, M.~Lemm, and I.~M.~Sigal, \textit{On Lieb-Robinson bounds for the Bose-Hubbard model}, Commun. Math. Phys. \textbf{394} (2022), 1011–1037
	%
	%
	%
	%
	%
	%
	%
	%
	%
	%
	%
	%
	%
	%
	\bibitem{Fossetal} M.~Foss-Feig, Z.-X.~Gong, C.W.~Clark, and A.V.~Gorshkov, \textit{Nearly-linear light cones in long-range interacting quantum systems}, Phys.\ Rev.\ Lett.\ \textbf{114} (2015), 157201 
	%
	%
	%
	%
	%
	%
	%
	%
	%
	%
	%
	
	
	
	
	
	%
	%
	%
	\bibitem{GoJe} S. Gol\'enia and T. Jecko,  \textit{A new look at Mourre's commutator theory}, Complex Anal. Oper. Theory \textbf{1}, No. 3, 399--422 (2007).	
	
	\bibitem{GGC}
	Z.~Gong, T.~Guaita, and J.I.~Cirac, \textit{Long-Range Free Fermions: Lieb-Robinson Bound, Clustering Properties, and Topological Phases
}, Phys. Rev. Lett. \textbf{130}, 070401 (2023)
	
	
	
	
	%
	\bibitem{H04}
	M.B.~Hastings, \emph{{Lieb-Schultz-Mattis in higher dimensions}} Phys.\ Rev.\ B \textbf{69} (2004), no.\ 10, 104431.
	
	\bibitem{H07}
	M.B.~Hastings, \textit{An area law for one-dimensional quantum systems}, J.\ Stat.\ Mech.: Theor.\ Exper.\ \textbf{2007} (2007), P08024
	%
	%
	%
	%
	%
	%
	%
	\bibitem{HW}
	M.B.~Hastings and X.G.~Wen, \emph{{Quasi-adiabatic continuation of quantum states: The stability of topological groundstate degeneracy and emergent gauge invariance}}, Phys.\ Rev.\ B. \textbf{72} (2005), 045141
	
	
	
	
	\bibitem{HerbstSkib}
	I.~Herbst and E.~Skibsted, \textit{Free channel Fourier transform in the long-range N-body problem} Journal d’Analyse Mathematique, 65(1) (1995), 297-332
	
	
	
	
	
	%
	
	
	
	\bibitem{HunSig1}  W. Hunziker and I.M. Sigal, \textit{Time-dependent scattering theory of n-body quantum systems}, Rev. Math. Phys., \textbf{12} (2000) 1033--1084
	
	
	
	
	
	
	
	\bibitem{HunSigSof}  W. Hunziker, I.M. Sigal and A. Soffer, \textit{Minimal escape velocities}. 
	Communications in Partial Differential Equations, {\bf 24} (1999),  2279-2295
	
	
	
	
	\bibitem{IngardenKossakowski}
	R.~S. Ingarden and A.~Kossakowski,
	\textit{On the connection of nonequilibrium information thermodynamics with
		non-{H}amiltonian quantum mechanics of open systems},
	\newblock {\em Ann. Physics}, \textbf{89}:451--485, 1975
	%
	%
	%
	%
	%
	%
	%
	%
	%
	%
	%
	%
	%
	%
	%
	%
	\bibitem{kossa}
	A.~Kossakowski,
	\textit{On necessary and sufficient conditions for a generator of a quantum
		dynamical semi-group},
	\newblock  Bull. Acad. Polon. Sci. S\'er. Sci. Math. Astronom. Phys.,
	\textbf{20}:1021--1025, 1972.
	
	%
	%
	
	\bibitem{KSV}
	T.~Kuwahara, T.~Van Vu, and  K.~Saito,  \textit{Optimal light cone and digital quantum simulation of interacting bosons}, arXiv preprint, arXiv:2206.14736.
	%
	%
	%
	%
	%
	
	
	
	
	
	
	
	
	\bibitem{LR}
	E.H.~Lieb, and D.W.~Robinson, \textit{The finite group velocity of quantum spin systems}, In Statistical mechanics, 425-431. Springer, Berlin, 1972 (see also E.H.~Lieb: Statistical Mechanics: Selecta of Elliott H. Lieb,  Bruno Nachtergaele, Jan Philip Solovej and J. Yngvason, Editors, Springer, 2014).
	
	
	
	
	
	
	
	
	%
	%
	%
	%
	%
	\bibitem{NRSS}
	B.~Nachtergaele, H.~Raz, B.~Schlein, and R.~Sims, \textit{Lieb-Robinson bounds for harmonic and anharmonic lattice systems}, Commun.\ Math.\ Phys.\ \textbf{286} (2009), no.\ 3, 1073-1098.
	%
	\bibitem{NS1} B.~Nachtergaele and R.~Sims, \textit{Lieb-Robinson bounds and the exponential clustering theorem}, Commun.\ Math.\ Phys.\ \textbf{265} (2006), no.\ 1, 119-130.
	
	\bibitem{NSY1}
	B.~Nachtergaele, R.~Sims, and A.~Young, \textit{Lieb-Robinson bounds, the spectral flow, and stability of the spectral gap for lattice fermion systems}, Mathematical Problems in Quantum Physics \textbf{717} (2018).
	
	
	\bibitem{NSY2}
	B.~Nachtergaele, R.~Sims, and A.~Young, \textit{Quasi-locality bounds for quantum lattice systems. I. Lieb-Robinson bounds, quasi-local maps, and spectral flow automorphisms} J.\ Math.\ Phys.\ \textbf{60}, no.\ 6, 061101.
	
	
	\bibitem{NVZ}
	B.~Nachtergaele, A.~Vershynina, V.~Zagrebnov, \emph{{Lieb-Robinson bounds and existence of the thermodynamic limit for a class of irreversible quantum dynamics}}, AMS Contemporary Mathematics \textbf{552} (2011) 161-175.	 
	%
	%
	%
	%
	
	
	\bibitem{OS} D.~Ou Yang and I.M.~Sigal, \textit{Approach to equilibrium in the von Neumann-Lindblad equation}, In preparation.
	
	\bibitem{Pou} 
	D.~Poulin, \emph{{Lieb-Robinson bound and locality for general Markovian quantum dynamics}}, Phys.\ Rev.\ Lett.\ \textbf{104} (2010), no.\ 19, 190401.
	
	
	
	
	
	
	
	%
	%
	%
	%
	\bibitem{SHOE}
	N.~Schuch, S.K.~Harrison, T.J.~Osborne, and J.~Eisert,
	\textit{Information propagation for interacting-particle systems}, Phys.\ Rev.\ A \textbf{84} (2011), 032309. 
	
	
	
	
	\bibitem{SigSof} I.M. Sigal and A. Soffer, \textit{Local decay and propagation estimates for time-dependent and time-independent Hamiltonians} 
	Preprint, Princeton Univ. (1988)  http://www.math.toronto.edu/sigal/publications/SigSofVelBnd.pdf.
	
	
	
	
	
	\bibitem{SigSof2} I.M. Sigal and A. Soffer, \textit{Long-range many-body scattering}, Invent. Math. \textbf{99} (1990) 115--143.
	
\bibitem{SZ} I.M. Sigal and J. Zhang, On propagation of information in quantum many-body systems. ArXiv 2023.	
	
	\bibitem{Skib} E. Skibsted, \textit{Propagation estimates for N-body Schr\"odinger operators}, Commun. Math. Phys. \textbf{142} (1992) 67--98.
	
	\bibitem{Stein70} E.M. Stein, \textit{Singular Integrals and Differentiability Properties of Functions}, Princeton University Press, Princeton, 1970.
	
	%
	%
	%
	%
	%
	\bibitem{Tranetal}
	M.C.~Tran, A.Y.~Guo, C.L.~Baldwin, A.~Ehrenberg, A.V.~Gorshkov, and A.~Lucas, \textit{The Lieb-Robinson light cone for power-law interactions}, Phys.\ Rev.\ Lett.\ \textbf{127} (2021), no.\ 16, 160401.
	%
	%
	%
                            \bibitem{WH}
                            Z.~Wang and K.R.~Hazzard, \textit{Tightening the Lieb-Robinson bound in locally interacting systems}, PRX
                             Quantum \textbf{1} (2020), 010303.
                          	\bibitem{YL}
                               	C.~Yin and A.~Lucas, \textit{Finite speed of quantum information in models of interacting bosons at finite density}, Phys. Rev. X 12, 021039
                                	
                                	
	
	
\end{thebibliography}

\end{document}